\documentclass[12pt]{article}

\RequirePackage[OT1]{fontenc}
\RequirePackage{amsthm,amsmath,amssymb,amsfonts,graphicx,subfigure,bm,fullpage,
titlesec}
\RequirePackage{hypernat}
\usepackage{setspace}

\numberwithin{equation}{section}
\theoremstyle{plain}
                          \newtheorem{thm}{Theorem}
                          
                          \newtheorem{prop}[thm]{Proposition}
\theoremstyle{remark}         
\theoremstyle{definition} \newtheorem{defn}{Definition}
                          \newtheorem{exam}{Example}

\newcommand\om{\Omega}
\newcommand\real{\mathbb{R}}
\newcommand\A{\mathcal{A}}
\newcommand\B{\mathcal{B}}

\newcommand\R{\mathcal{R}}
\newcommand\J{\mathcal{J}}
\newcommand\bj{\bm{J}}
\newcommand\bx{\bm{x}}

\newcommand\bn{\bm{n}}
\newcommand\balp{\bm\alpha}
\newcommand\blam{\bm\lambda}
\newcommand\bthe{\bm\theta}

\newcommand\Var{{\rm Var}}
\newcommand\bX{\bm{X}}
\newcommand\by{\bm{y}}

\titleformat{\section}{\large\bfseries}{\thesection}{1em}{} 
\titleformat{\subsection}{\normalsize\bfseries}{\thesubsection}{1em}{} 

\providecommand{\MR}[1]{} 

\def\I{{\mathbf{1}}}

\usepackage{varioref}		    
\labelformat{section}{Section~#1}   
\labelformat{figure}{Figure~#1}
\labelformat{table}{Table~#1}

\begin{document}

\title{\Large Coupling optional P\'olya trees and the two sample problem}
\author{\large Li Ma\footnote{Department of Statistics, 390 Serra Mall, Stanford,
    CA 94305 / EM: \texttt{ma2@stanford.edu}} 
  \and \large Wing H. Wong\footnote{Department of Statistics, 390 Serra
    Mall, Stanford, CA 94305 / EM: \texttt{whwong@stanford.edu}} 
}
\doublespacing

\maketitle
\vspace{-1em}

\begin{abstract}
  Testing and characterizing the difference between
  two data samples is of fundamental interest in statistics. Existing methods such as Kolmogorov-Smirnov and
  Cramer-von-Mises tests do not scale well as the dimensionality
  increases and provide no easy way to characterize the difference
  should it exist. In this work, we propose a theoretical
  framework for inference that addresses
  these challenges in the form of a prior for Bayesian nonparametric analysis.
  The new prior is constructed based on a random-partition-and-assignment procedure
  similar to the one that defines the standard optional P\'olya tree
  distribution, but has the ability to generate multiple random
  distributions jointly. These random probability
  distributions are allowed to ``couple'', that is to have the same
  conditional distribution, on subsets of the sample
  space. We show that this ``coupling optional P\'olya tree'' prior
  provides a convenient and effective way for
  both the testing of two sample difference and the learning of the
  underlying structure of the difference. In addition, we discuss
  some practical issues in the computational implementation of this prior and
  provide several numerical examples to demonstrate its work. 
  \end{abstract}

\section{Introduction}\label{sec:intro}
Two sample comparison is a fundamental problem in
statistics. With two samples of data at hand, one
often wants to answer the question---``Did these two samples
come from the same underlying distribution?'' In other words, one
is interseted in testing the null hypothesis that
the two data samples were generated from the same
distribution. Moreover, in the presence of evidence for deviation
between the two samples, one often hopes to learn
the structure of such difference in order to understand, for example, what factors
could have played a role in causing the difference. Hence two sample comparison is
interesting both as a hypothesis testing problem and as a data mining
problem. In this work, we consider the problem from both
aspects, and develop a Bayesian nonparametric approach that can serve both the
testing and the learning purposes.

Nonparametric hypothesis testing for two sample difference has a long
history and rich literature, and many methods have been proposed. Some well-known
examples include Wilcoxon test \cite[p.243]{lehmannandromano:2006}, Kolmogorov-Smirnov test
\cite[pp. 392--394]{chakravartistat} and Cramer-von-Mises test \cite{anderson:1962}. Recently, this 
problem has also been investigated
from a Bayesian nonparametric perspective using a P\'olya tree prior
\cite{holmes:2009}.

Despite the success of these existing testing methods for
one-dimensional problems, two sample comparison in multi-dimensional spaces remains a challenging task. 
A basic idea for many existing methods is to
estimate the two underlying distributions, and then use a distance metric to measure the
dissimilarity between the two estimates. Tests such as
Kolmogorov-Smirnov (K-S)
and Cramer-von-Mises (CvM) fall into this category. However, reliably
characterizing distributions in multi-dimensional problems, if
computationally feasible at all, often requires a
prohibitively large number of data points. With even just a moderate
number of dimensions, the estimated distributional distance is
often highly variable or biased.
This ``curse of dimensionality'' demonstrates itself in the Bayesian
setting as well. 
This is true even when the underlying difference is
structurally very simple and can be accounted for by a relatively
small number of dimensions in the space.

One general approach to dealing with the curse of dimensionality when
characterizing distributions in a multi-dimensional space is to
learn from the data a partition of the
space that best reflects the underlying
structure of the distribution(s). A good partition of the space overcomes the sparsity of the
data by placing true neighbors together, and it reduces
computational burden by allowing one to focus on the relevant blocks
in the space. Hence it can be very helpful in multi-dimensional, and
especially high-dimensional,
settings to incorporate the learning of
a representative partition of the space into the inference procedure.
Wong and Ma \cite{wongandma:2010} adopted this idea and introduced the optional
P\'olya tree (OPT) prior as such a method under the Bayesian
nonparametric framework. Through
optional stopping and randomized splitting of the sample space, a 
recursive partitioning procedure is incorporated into the parametrization of
this prior, thereby allowing the data to suggest
 parsimonious divisions of the space. The OPT prior, like other
 existing Bayesian
nonparametric priors, deals with only one data
sample, but as we will demonstrate in this paper, similar ideas can be
utilized for problems involving more than one sample as well.

Besides the difficulty in handling multidimensional problems,
existing nonparametric methods for two sample comparison are also
unsatisfactory in that they 
provide no easy way to learn the underlying
structure of the difference should it exist. Tests such as K-S and CvM
provide statistics with which to test for the existence of a
difference, but does not allow one to characterize the
difference---for example what variables are involved in the difference
and how. One has to resort to methods such as logistic regression that
rely on strong modelling
assumptions to investigate such structure. Similarly, Bayes factors
computed using nonparametric priors such as Dirichlet
process mixture and the P\'olya tree prior also shed no light on where the
evidence for difference has arisen.   

In this work, we introduce a new prior called ``coupling optional P\'olya
tree'' (co-OPT) designed for Bayesian nonparametric inference on the two sample
problem. This new prior {\em jointly} generates two random
distributions through a random-partitioning-and-assignment procedure
similar to the one that gives rise to the OPT prior
\cite{wongandma:2010}. The co-OPT framework allows
both hypothesis testing on the null hypothesis and posterior learning
of the distributional difference in terms of a partition of the
space that ``best'' reflects the difference structure. The ability to make
posterior inference on a partition of the space also enhances the
testing power for multi-dimensional problems.

This paper is organized as follows. In Section~2 we
review the construction of the OPT distribution. In Section~3 we
generalize the definition of the OPT distribution by replacing the
``uniform base measure'' (defined later) with a general absolutely continuous distribution, and
show that this generalized prior can be used for investigating the
goodness-of-fit of the data to the base distribution. 
In Section~4 we introduce the co-OPT prior and show how Bayesian inference
can be carried out using this prior. In addition, we discuss the practical issues in
implementing inference using this prior. In Section~5 we provide
several numerical examples 
to illustrate inference on the two sample comparison problem using
this prior and compare it to other methods. Then in Section~6 we present a method for
inferring two common distributional distances, $L_1$ and
Hellinger, between the two sample distributions 
using a co-OPT prior and provide two more numerical examples. Section~7 concludes with a few remarks.  

We close this introduction with a few words on the recent development
in the Bayesian nonparametric literature on related topics. In the
past 15 years, several methods have been proposed for testing
the one sample goodness-of-fit, in particular, for non-parametric alternatives against a
parametric null. For some examples see \cite{florens1996, carota1996,
  berger2001, basu2003, hanson2006, mcvinish2009, tokdar2010}. As for
two sample comparison, Holmes
{\em et. al.} \cite{holmes:2009} introduced a way to compute the
Bayes factor for testing the null through  
the marginal likelihood of the data with P\'olya tree priors. Under the
null, they model
the two samples to have come from a single random measure distributed
as a P\'olya tree, while under
the alternative from two separate P\'olya tree distributions. In contrast, our
new prior allows the two distributions to be generated {\em jointly}
through one prior even when they are different. It
is this joint generation that allows both the testing of the
difference and the learning of the structure simultaneously. There are
other approaches to joint modeling of multiple distributions in the
Bayesian nonparametric literature. For example, one idea is to
introduce dependence structure into
Dirichlet processes \cite{maceachern1999}. For some notable examples
see \cite{muller2004, teh2006,
  griffin2006}, among many others. Compared to these methods based on
Dirichlet processes, our method, based on the optional P\'olya tree,
allows the resolution of the inference to be adaptive to the data
structure and handles the sparsity in multidimensional settings using
random partitioning \cite{wongandma:2010}. Moreover, our method allows direct inference on the distributional difference without relying on
inferring the two distributions {\em per se}, making it particularly
suited for comparison across multiple samples. This point will be
further discussed in Section 4.2 and illustrated in the
examples given in Sections 5 and 6.

\section{Optional P\'olya trees and Bayesian inference}\label{sec:OPT}
Wong and Ma \cite{wongandma:2010} introduced the optional P\'olya tree (OPT) distribution as an
extension to the P\'olya tree prior that allows optional stopping and
randomized partitioning of the sample space $\om$, where $\om$ is
either finite or a rectangle in an Euclidean space. One can think of
this prior
as a procedure for generating random probability measures on $\om$
that consists of two components---(1) random partitioning of the space and
(2) random probability assignment into the parts of the space produced
by the partitioning.

We first review how the OPT prior randomly partitions the
space. Let $\R$ denote a partition rule function which, for any
subset $A$ of $\om$, defines a number of ways to partition $A$ into a
finite number of smaller sets. 
For example, for $\om=[0,1] \times [0,1]$, the (coordinate-wise) diadic split rule $\R$ is that
$\R(A)$=\{ splitting $A$ in the middle of the range of one of the
coordinates $x_j$, $j=1$ or 2\} if $A$ is a
non-empty rectangle and = $\emptyset$ otherwise. We call a rule
function $\R$ finite if $\forall A \subset \om$, the number of
possible ways to
partition $A$, $M(A)$, as specified by $\R$, is finite. In the
rest of the paper, we will only consider
finite partition rules. Let $K^j(A)$ be the number of children specified by the
$j$th way to partition $A$ under $\R(A)$, and let $A^j_i$ denote the
$i$th child set of $A$ in that way of partitioning. That is,
$A=\cup_i^{K^j(A)}A^j_i$ for $j=1,2,\ldots, M(A)$. We can write $\R(A)$ 
as
\begin{align*}
\R(A) & =\{ \{A^1_1, A^1_2, \ldots A^1_{K^1}\}, \{A^2_1, A^2_2,
\ldots, 
A^2_{K^2}\},\ldots  \{A^{M}_1, A^{M}_2, \ldots
A^{M}_{K^{M}}\} \}=\left\{ \{ A^j_i\}_{i=1}^{K^j} \right\}_{j=1}^{M},
\end{align*}
where for simplicity we suppressed notation by writing $M$ for $M(A)$ and $K$ for $K(A)$.

A partition rule function $\R$ does not specify any particular partition on
$\om$ but rather a collection of possible partitions over which one
can draw random samples. The OPT prior samples from this collection of partitions
in the following sequential
way. Starting from the whole space $A=\om$. If $M(A)=0$, then $A$ is
not divisible under $\R$ and we call
$A$ an atom (set). In this case the partitioning of $A$ is completed. If
$M(A)>0$, that is, $A$ is divisible, then a Bernoulli($\rho(A)$) random
variable $S(A)$ is drawn. If $S(A)=1$, we stop partitioning
$A$. Hence $S(A)$ is called the stopping variable for $A$, and
$\rho(A)$ the stopping probability.
If $S(A)=0$, $A$ is divided in the $J(A)$th of the $M(A)$ available ways
for partitioning $A$ under $\R(A)$, where $J(A)$ is a random variable
taking values $1,2,\ldots, M(A)$ with probabilities $\lambda_1(A),
\lambda_2(A), \ldots, \lambda_{M(A)}(A)$ respectively, and $\sum_{j=1}^{M(A)}\lambda_j(A)=1$.   
$J(A)$ is hence called the (partition) selector variable, and $\blam(A)=(\lambda_1(A),
\lambda_2(A), \ldots, \lambda_{M(A)}(A))$ the (partition) selector
probabilities. If $J(A)=j$, we partition $A$ into $\{A^j_1, A^j_2,
\ldots A^j_{K^j(A)}\}$, and then apply the same procedure to each of
the children. In addition, if $A$ is reached from $\om$ after $k$
steps (or levels) of recursive partitioning, then we say that $A^j_i$
is reached after $k+1$ steps (or levels) of recursive partitioning. (To complete this inductive
definition, we say that the
space $\om$ is reached after 0 steps of recursive partitioning.) The
recursive partitioning procedure naturally gives rise to a tree
structure on the sample space. For this reason, we shall also refer to the
sets $A$ that arise during the precedure as (tree) nodes.

The first question that naturally arises is whether this
sequential procedure will eventually ``stop'' and produce a well
defined partition on $\om$. Given that the
stopping probability $\rho(A) > \delta$ for some $\delta$ and all $A$,
this is indeed true 
in the following sense. If we let $\mu$ be the natural measure on
$\om$---the Lebesgue measure if $\om$ is a rectangle in
an Euclidean space or the counting measure if $\om$ is finite, then
$\mu(T_1^k) \rightarrow 0$ with probability 1, where $T_1^k$ is the part of $\om$
that is still not stopped after $k$ steps of recursive partitioning. In other words, the
partitioning procedure will stop
almost everywhere on $\om$.

The second component of the OPT prior is random probability
assignment. The prior assigns probability mass into
the randomly generated parts of the space in the following manner. Starting from $A=\om$, assign
$Q(A)=1$ total probability to $A$. If $A$ is stopped or is an atom, then let the
conditional distribution within $A$ be uniform. That is, $Q(\cdot |
A)=u(\cdot |A)$, where $u$ denotes the uniform density (w.r.t.\!
$\mu$) and this completes the probability assignment on
$A$. If instead $A$ has children $\{A^j_1, A^j_2,
\ldots A^j_{K^j(A)}\}$, (this occurs when $S(A)=0$ and $J(A)=j$,) a random vector $(\theta^j_{1}(A), \theta^j_{2}(A),
\ldots, \theta^j_{K^j(A)}(A))$ on the $K^j(A)-1$ dimensional
simplex is drawn from
a Dirichlet($\alpha^j_{1}(A),\alpha^j_{2}(A),\ldots,\alpha^j_{K^j(A)}(A)$)
distribution, and we assign to each child $A^j_i$ probability mass
$Q(A^j_i)=Q(A)\theta^j_i(A)$. We call $\bthe^j(A)= (\theta^j_{1}(A), \theta^j_{2}(A),
\ldots, \theta^j_{K^j(A)}(A))$ the (probability) assignment
vector, and $\balp^j(A)=(\alpha^j_{1}(A),\alpha^j_{2}(A),\ldots,\alpha^j_{K^j(A)}(A))$
the pseudo-count parameters. Then we go to the next level and assign
probability mass within each of the children in the same way.

Theorem~1 in \cite{wongandma:2010} shows that if $\rho(A) > \delta$ for
some $\delta>0$ and all $A$, then with probability 1 this random partitioning and
assignment procedure will give rise to a probability measure $Q$ on
$\om$ that is absolutely continuous with respect to $\mu$. This random
measure $Q$ is said to have an OPT distribution with (partition rule
$\R$ and) parameters $\rho$, $\blam$ and $\balp$, which can be written
as $OPT(\R;\rho,\blam,\balp)$. In addition, Wong and Ma
\cite{wongandma:2010} also show that under mild conditions, this prior
has large
support---any $L_1$ neighborhood of an absolutely continuous distribution (w.r.t. $\mu$) on
$\om$ has positive prior probability.

Two key features of the prior are demonstrated from the above constructive
description. The first is
self-similarity. If a set $A$ is
reached as a node during the recursive partitioning procedure, then the continuing partitioning and assignment
{\em within} $A$, which specifies the conditional distribution on $A$,
is just an OPT procedure with $\om = A$. The second
feature is the prior's implicit hierarchical structure. To see this,
we note that the random distribution that arises from such a prior is
completely determined by the partition and assignment variables $S$,
$J$, and $\bthe$, while the prior parameters
$\rho$, $\blam$ and $\balp$ specify the distributions of these
``middle'' variables. 

These two features allow one to write down a recursive formula for the
likelihood under a random distribution arising from such a prior. To
see this, first let $Q$ (with density $q$) be a distribution arising from an 
$OPT(\R;\rho,\blam,\balp)$ distribution, and for $A \subset \om$, let
$q(\cdot |A)$ be the conditional density on $A$. Let $S$, $J$, and $\bthe$ be the
corresponding partition and assignment variables for $Q$ (or $q$). Suppose one has $n$ i.i.d. observations,
$x_1$, $x_2$, \ldots, $x_n$, on $\om$ from $q(\cdot | \om)$. Define
\[
\bx(A) = \{ x_1, x_2, \ldots, x_n \} \cap A,
\]
the observations falling in $A$, and $n(A)=|\bx(A)|$, the number
of observations in $A$. Then for any node $A$ reached in the
recursive partitioning process determined by the $S$ and $J$ variables, the likelihood of 
observing $\bx(A)$ conditional on $A$ is
\begin{equation}
q\left(\bx(A) | A \right)= S u\left( \bx(A)|A
\right)+(1-S)\left(\prod_{i=1}^{K^{J}}\left(\theta_i^J\right)^{n(A_i^J)}\right)
\left(\prod_{i=1}^{K^J}q\left(\bx\left(A_i^J\right)\big|A_i^J\right)\right),
\label{eq1}
\end{equation}
where $u(\bx(A)|A)=\frac{1}{\mu(A)^{n(A)}}$ is the likelihood under
the uniform distribution on $A$,
$S=S(A)$, $J=J(A)$, $K^J=K^{J(A)}(A)$, and
$\theta_i^J=\theta_i^{J(A)}(A)$. (Note that for this formula to hold we need
to define
$q\left(\emptyset|A \right) := 1$.) From now on we will always suppress the ``(A)''
notation for the random variables and the parameters where
this adds no confusion. Similarly, we will use $q(\bx|A)$
and $u(\bx|A)$ to mean $q(\bx(A)|A)$ and $u(\bx(A)|A)$,
respectively.

Integrating out $S$, $J$, and $\bthe$ in $\eqref{eq1}$, we get the
corresponding recursive representation of the marginal likelihood
\begin{equation}
P(\bx | A) =
\rho u(\bx|A)+(1-\rho)\sum_{j=1}^M\lambda_j\,\frac{D(\bn^j+\balp^j)}{D(\balp^j)}\prod_{i=1}^{K^j}P\left(\bx|A_i^j\right), 
\label{eq2}
\end{equation}
where $P(\bx|A)=P(\bx(A)|A)$, $\bn^j=\bn^j(A)=(n(A^j_1),n(A^j_2), \ldots, n(A^j_{K^j(A)}))$, and $D(\bm{t})=\Gamma(t_1)\dots\Gamma(t_k)/\Gamma(t_1+\dots+t_k)$.
Wong and Ma \cite{wongandma:2010} provide terminal conditions so that
\eqref{eq2} can be used to compute the marginal
likelihood conditional on $A$, $P(\bx|A)$, for all potential tree nodes
$A$ determined by $\R$.

The final result we review in the section is the conjugacy of the OPT
prior. More specifically, given the i.i.d. observations $\bx$, the
posterior distribution of $Q$ is again an OPT distribution with
\begin{enumerate}
\item Stopping probability:
$\quad \rho(A|\bx)=\rho(A)u(\bx|A)\big/P(\bx|A)$

\item Selection probabilities:
\begin{equation*}
\lambda_j(A|\bx)\propto\lambda_j(A)\,\frac{D(\bn^j+\balp^j)}{D(\balp^j)}\prod_{i=1}^{K^j}P\left(\bx|A_i^j\right)
\qquad \text{for $j=1,\ldots,M(A)$}
\end{equation*}
\item Probability assignment pseudo-counts:
$\quad \alpha^j_i(A|\bx) = \alpha^j_i(A) + n(A^j_i)$\\
for $j=1,\ldots,M(A)$ and $i=1,2,\ldots, K^j(A)$
\end{enumerate}
where again $A$ is any potential node determined by
the partition rule function $\R$ on $\om$.

\section{Optional P\'olya trees and 1-sample ``goodness-of-fit''}\label{sec:gOPT}
In the constructive procedure for an OPT distribution described above,
whenever a node $A$ is stopped, the conditional distribution within
it is generated from that of a baseline distribution, namely the uniform
$u(\cdot |A)$. For this reason, we say that the collection of
conditional uniform
distributions, $\{ u(\cdot |A) : A \text{ is a potential node under }
\R \}$, are the {\em local} base
  measures. With uniform local base measures, the stopping
probability $\rho$ for a region $A$ represents the
probability that the distribution is ``flat'' within $A$. Accordingly,
the posterior OPT concentrates probability mass around
partitions that best captures the ``non-flatness'' in the density of the
data distribution. Such a partitioning criterion is most natural in
the context of density estimation.

One can extend the original OPT construction by adopting different
local base measures or stopping criteria for the nodes. More
specifically, we can replace $u(\cdot|A)$ with any absolutely continuous
measure $m^A(\cdot)$ on node $A$ in the probability assignment
step. That is, when a tree node $A$ is
stopped, we let the conditional distribution in $A$ be
$m^A(\cdot)$. With this generalization, the recursive constructive procedure for
the OPT distribution and the recipe for Bayesian inference
described in the previous section still follow
through. 

One choice of the $m^A$ measures
is of particular interest. Specifically, we can let
$m^A(\cdot)=q_0(\cdot |A)$ for some 
absolutely continuous distribution $Q_0$ with density $q_0$ on
$\om$. For this special case, we have the following definition.

\begin{defn}
The random probability measure $Q$ that arises from the
random-partitioning-and-assignment (RPAA) procedure described in the previous
section, with $u$ replaced by $q_0$, the density (w.r.t.\! $\mu$) of an absolutely
continuous distribution $Q_0$, is
said to have an optional P\'olya tree distribution on $\R$ with
parameters
$\blam$, $\balp$, $\rho$, and (global) base measure (or distribution) $Q_0$. We denote this
distribution by
$OPT(\R; \blam, \balp, \rho; Q_0)$.
\end{defn}

The next theorem shows that by choosing an appropriate
partitioning rule $\R$ and/or suitable pseudocount parameters
$\balp$, one can enforce the random distribution $Q$ to ``center
around'' the base measure $Q_0$.

\begin{thm}
If $Q \sim OPT(\R; \rho,\blam,\balp; Q_0)$, where
$\delta<\rho(A)$ for some $\delta$ and all potential tree nodes $A$, then $\forall$ Borel set $B$, 
\[
EQ(B)=Q_0(B),
\]
provided that for all $A$, $j=1,2,\ldots,M(A)$ and $i=1,2,\ldots,K^j(A)$,
we have 
\[ \alpha^j_i(A)/\sum_{h=1}^{K^j(A)}\alpha^j_h(A)=Q_0(A^j_i)/Q_0(A).\]
\label{thm:mean}
\end{thm}
\vspace{-2em}

\begin{proof} 
See supplementary materials.
\end{proof}
\noindent Remark: If we have equal pseudocounts, that is,
$\alpha^j_1(A)=\alpha^j_2(A)=\cdots=\alpha^j_{K^j(A)}(A)$ for all
potential nodes $A$ and all $j$, then the condition for the theorem becomes
$Q_0(A^j_i)/Q_0(A)=1/K^j(A)$. Therefore one can choose a partition
rule $\R$ on $\om$ based on the base measure to center the prior.

Bayesian inference using the OPT prior with a general base measure can be carried out
just as before, provided we replace $u(\bx|A)$ replaced by $q_0(\bx|A)$ everywhere. 
An important fact is that a random distribution
with this prior has positive probability to be {\em
  exactly} the same as the base distribution. 
 Therefore, one can think of the inferential procedure for the OPT
prior as a sequence of
recursive comparison steps to the base measure. More specifically, the partitioning decision on
each node $A$ is determined by comparing the conditional likelihood
of the data within $A$ under $Q_0$ to the composite of $M(A)$
composite alternatives. The partition of
each node $A$ stops when the observations in $A$ ``fits'' the structure
of the base measure, and the posterior values of the partitioning
variables capture the
discrepancy, if any, between the data and the base. Consequently, this
framework can be used to recursively test for 1-sample {\em
  goodness-of-fit} and to learn the structure of any potential ``misfit''. For each node $A$, the
posterior stopping
probability is the probability that the data distribution 
coincides with the base
distribution conditional on $A$. In particular, the posterior
stopping probability for the whole space $\om$, $\rho(\om)$, measures how
well the observed data fit the base overall. The posterior
values of the other partitioning and pseudocount variables reflect where and how the data distribution
differs from the base.

\section{Coupling optional P\'olya trees and two sample comparison}\label{sec:co-OPT}
In this section we consider the case when two i.i.d. samples are observed and
one is interested in testing and characterizing the potential difference between
the underlying distributions. From now on, we let $Q_1$ and $Q_2$, with densities $q_1$ and $q_2$, be the two 
distributions from which the two samples have come from. 

\subsection{Coupling optional P\'olya trees}
A conceptually simple way to compare $Q_1$ and $Q_2$ is to proceed
in two steps---first estimate
the two distributions {\em separately}, and then use some distance metric to quantify the
difference. For example, one can place an OPT prior
on each of $Q_1$ and $Q_2$ and use the posteriors to estimate the
densities \cite{wongandma:2010}. (Other density estimators can also be
used for this purpose.) With the density estimates available,
one can then compute standard distance metrics such as $L_1$, and in
turn use this as a statistic for testing the difference. (This approach
provides no easy way to characterize how the two distributions are different.)

However, this two-step method is
undesirable in multidimensional, and especially high-dimensional, settings. The main reason
is that reliably estimating
multidimensional distributions 
is a very difficult problem, and in fact often a much harder
problem than comparing distributions. This difficulty in turn
translates into either high variability or large bias in the 
distance estimates, and thus low statistical power. Using this approach, one is essentially making
inference on the distributional difference {\em indirectly}, through the inference on
a large number of parameters that characterize the two distributions
{\em per se} but have little to do with their difference. 

Following this reasoning, it is favorable to make {\em direct} inference on
``parameters'' that capture the distributional difference. But such
direct inference requires (from a Bayesian perspective) that  the two
distributions be generated from a {\em joint}
prior. This prior should be so designed that in the corresponding joint posterior,
information regarding the distributional difference can be extracted
directly. We next introduce such a prior. 

Our proposed method for generating the two distributions $Q_1$ and $Q_2$ is again based on
a procedure that randomly partitions the space $\om$ and assigns probability
masses into the parts, similar to the one that defines the OPT prior. What differs from the procedure for the OPT is that we add in an extra random component---the conditional coupling of the two measures $Q_1$ and $Q_2$ within the tree nodes. We next explain this construction in detail. 
Starting from the whole space $A=\om$, we draw a random variable 
\vspace{-1em}

\[ C(A)\sim Bernoulli(\gamma(A)),\] 
which we call the coupling variable. If $C(A)=1$, then we force $Q_1$ and $Q_2$ to be coupled conditional on $A$---that is, $Q_1(\cdot |A)=Q_2(\cdot |A)$---and we achieve this by generating a common conditional distribution from a stanford OPT on $A$. That is
\[ Q_1(\cdot |A) = Q_2(\cdot | A) \sim OPT_{|A}(\R; \rho, \blam^{b}, \balp^{b}), \] 
where the ``b'' superscript stands for ``base'', and the ``$|A$" notation should be understood as the restriction to $A$ of the partition rule $\R$, the stopping variables $\rho$, the partition selector variables $\blam^b$, and the assignment pseudo-count variables $\balp^b$.
(For $A=\om$, there is no restriction.) If $C(A)=0$, then we draw a partition selector variable 
\[J(A) \in \{1,2,\ldots, n\} \quad \text{with $P(J(A)=j)=\lambda^{j}(A).$}\]
If $J(A)=j$, then we partition $A$ under the $j$th way according to $\R(A)$. Then draw two independent assignment vectors 
\begin{align*} 
\bthe^j_1(A)=(\theta^j_{11}(A), \theta^j_{12}(A), \ldots, \theta^j_{1K^j(A)}(A)) &\sim Dirichlet(\alpha^j_{11}(A),\alpha^j_{12}(A),\ldots,\alpha^j_{1K^j(A)}(A))\\
\bthe^j_2(A)=(\theta^j_{21}(A), \theta^j_{22}(A), \ldots, \theta^j_{2K^j(A)}(A)) &\sim Dirichlet(\alpha^j_{21}(A),\alpha^j_{22}(A),\ldots,\alpha^j_{2K^j(A)}(A)),\\
\end{align*}
\vspace{-4em}

\noindent and let 
\[
Q_1(A^j_i)=Q_1(A)\theta^j_{1i}(A) \quad  \text{and} \quad
Q_2(A^j_i)=Q_2(A)\theta^j_{2i}(A)
\]
for each child $A^j_i$ of $A$. We call $\bthe^j_1(A)$ and
$\bthe^j_2(A)$ the assignment vectors for $Q_1$ and $Q_2$ (in the {\em
  uncoupled} state). Then we go down one level and repeat the entire procedure for each $A^j_i$, starting from the drawing of the coupling variable.

Again, the first natural question to ask is whether this procedure will
actually stop and give rise to two random probability measures
$(Q_1,Q_2)$. The answer is positive under very mild conditions, and
this is formalized in Theorem~\ref{thm:coOPT}. 
The statement of the theorem uses the notion of ``forced coupling'', which is similar
to the idea of ``forced stopping'' used in the proof of
Theorem~\ref{thm:mean} and which we describe next. Let $(Q_1^{(k)},Q_2^{(k)})$ denote the pair of random distributions arising
from the above random-partitioning-and-assignment procedure with
forced coupling after $k$-levels of recursive partitioning. That is,
if after $k$ levels of partitioning a node $A$ is reached and
the two measures are not coupled on it, then force them to couple
on $A$ and generate $Q_1^{(k)}(\cdot |A)=Q_2^{(k)}(\cdot|A)$ from
$OPT_{|A}(\R;\rho,\blam^b,\balp^b)$. We do this for all such nodes to get $(Q_1^{(k)}, Q_2^{(k)})$.
\begin{thm}
In the random-partitioning-and-assignment procedure for generating a
pair of measures described
above, if $\gamma(A),\rho(A)>\delta$ for some $\delta>0$ and all potential nodes $A$
defined by the partition rule $\R$, then with probability 1,
$(Q_1^{(k)},Q_2^{(k)})$ converges to a pair of absolutely continuous
(w.r.t. $\mu$) random probability measures $(Q_1,
Q_2)$ in the following sense.
\[
sup_{E \in \B} |Q_1^{(k)}(E)-Q_1(E)| + |Q_2^{(k)}(E)-Q_2(E)|
\rightarrow 0,
\] 
where $\B$ is the collection of Borel sets.
\label{thm:coOPT}
\end{thm}

\begin{defn}
This pair of random probability measures ($Q_1,Q_2$) is
said to have a coupling optional P\'olya tree (co-OPT) distribution with
partition rule $\R$,
coupling parameters $\blam$, $\balp_1$, $\balp_2$, $\gamma$, and base parameters $\blam^b$, $\balp^b$, $\rho$, and can be written as
{\it co-OPT}$(\R; \blam, \balp_1, \balp_2, \gamma; \blam^b, \balp^b, \rho)$.
\end{defn}

\begin{proof}[Proof of Theorem~\ref{thm:coOPT}]
See supplementary materials.
 \end{proof}
Similar to the OPT prior, the co-OPT distribution has
large support under the $L_1$ metric. This is formulated in
the following theorem.
\begin{thm}
Let $\om$ be a bounded rectangle in $\real^p$. Suppose that the
condition of Theorem \ref{thm:coOPT} holds along with the following
two conditions:
\begin{itemize}
\item[(1)] For any $\epsilon>0$, there exists a partition of the sample space allowed under $\R$,
$\om = \cup_{i=1}^I A_i$,  such that the diameter of each node $A_i$ is less then $\epsilon$.
\item[(2)] The
coupling probabilities $\gamma(A)$, stopping probabilities $\rho(A)$, coupling selector
probabilities $\lambda^j(A)$, base selection probabilities
$\lambda^{b}_{j}(A)$, as well as the assignment probabilities  
$\alpha_{1i}^j(A)/(\sum_l\alpha_{1l}^j(A))$,
$\alpha_{2i}^j(A)/(\sum_l\alpha_{2l}^j(A))$, and
$\alpha_{i}^{bj}(A)/(\sum_l\alpha_{l}^{bj}(A))$ for all $i$, $j$ and all potential elementary
regions are uniformly bounded
away from $0$ and $1$.
\end{itemize}
Let
$q_1=dQ_1/d\mu$ and $q_2=dQ_2/d\mu$, then for any two density functions
$f_1$ and $f_2$, and any $\tau>0$, we have
\begin{equation*}
P\left(\int|q_1(x)-f_1(x)|d\mu<\tau \text{ and }
  \int|q_2(x)-f_2(x)|d\mu<\tau \right)>0.
\end{equation*}
\label{thm:coOPT_large_support}
\end{thm}
\vspace{-3em}

\begin{proof}
See supplementary materials.
\end{proof}

\subsection{Bayesian inference on the two sample problem using co-OPT prior}
We next show that the co-OPT prior is conjugate and introduce the
recipe for making inference on the two sample problem using this
prior. Suppose $(Q_1,
Q_2)$ is distributed as {\it co-OPT}$(\R; \blam, \balp_1,\balp_2,\gamma; \blam^b, \balp^b, \rho)$, and we observe two i.i.d.\! samples
$\bx_1=(x_{11},x_{12}, \ldots, x_{1n_1})$ and 
$\bx_2=(x_{21},x_{22}, \ldots, x_{2n_2})$ from $Q_1$ and
$Q_2$ respectively. For a node $A$
reached during the random partitioning steps in the
generative procedure of $(Q_1,Q_2)$, let
$\bx_1(A)=\{x_{11},x_{12}, \ldots, x_{1n_1}\} \cap A$ and $\bx_2(A)=\{ x_{21},x_{22}, \ldots, x_{2n_2}
\}\cap A$ be the
observations from the two samples in $A$, and let
$n_1(A)=|\bx_1(A)|$ and $n_2(A)=|\bx_2(A)|$ be the sample sizes in
$A$. As before, we let $q_1$ and $q_2$ denote the densities of the two
distributions and let $q^{A}_0$ denote the density of the
random local base measure $Q^A_0$.

The likelihood of $\bx_1(A)$ on $A$ under $q_1(\cdot |A)$ and that for
$\bx_2(A)$ under $q_2(\cdot |A)$ are
\begin{align}
\left\{ \begin{array}{ll} q_1(\bx_1|A)=C q^{A}_0(\bx_1) +
    (1-C)\prod_{i=1}^{K^{J}}(\theta^{J}_{1i})^{n_1(A^J_i)}q_1(\bx_1|A^J_i)\\\\
    q_2(\bx_2|A)=C q^{A}_0(\bx_2) +
    (1-C)\prod_{i=1}^{K^{J}}(\theta^{J}_{2i})^{n_2(A^J_i)}q_2(\bx_2|A^J_i)\end{array}
\right.
\label{eq:coOPT_mlik}
\end{align}
where we have again suppressed the ``(A)'' notation for $C(A)$, $J(A)$,
$K(A)^{J(A)}$, $\theta^{J(A)}_{1i}(A)$, $\theta^{J(A)}_{2i}(A)$,
$\bx_1(A)$ and $\bx_2(A)$. The joint likelihood of observing
$\bx_1(A)$ and $\bx_2(A)$ conditional on $A$ is
\begin{align}
q_1(\bx_1|A)q_2(\bx_2|A)=C q^{A}_0(\bx_1,\bx_2) +
(1-C)\prod_{i=1}^{K^{J}}(\theta^{J}_{1i})^{n_1(A^J_i)}(\theta^{J}_{2i})^{n_2(A^J_i)} q_1(\bx_1|A^J_i)
q_2(\bx_2 |A^J_i),
\label{eq:coOPT_lik}
\end{align}
\vspace{-2em}

\noindent where $q^{A}_0(\bx_1, \bx_2)=q^{A}_0(\bx_1)q^{A}_0(\bx_2)$ is the
standard OPT likelihood for the combined sample $\bx(A)=(\bx_1(A),
\bx_2(A))$ on $A$ given by \eqref{eq1}.
Integrating out $q^{A}_0$, $C$, $J$, $\bthe^{J}_1$ and $\bthe^{J}_2$ from
\eqref{eq:coOPT_lik}, we get the conditional marginal likelihood
\begin{align}
P(\bx_1, \bx_2 |A) 
&= \gamma P_0(\bx_1, \bx_2 |A) + (1-\gamma)
\sum_{j=1}^{M}\lambda_j \frac{D(\bn_1^j+\balp_1^j)D(\bn_2^j+\balp_2^j)}{D(\balp_1^j)D(\balp_2^j)}\prod_{i=1}^{K^j}P(\bx_1,\bx_2|A^j_i),  
\label{eq:coOPT}
\end{align}
\vspace{-2em}

\noindent where $\bn_h^j=(n_h(A^j_1), n_h(A^j_2), \ldots, n_h(A^j_{K^j})$ and
$\balp_h^j=(\alpha^j_{h1}(A),\alpha^j_{h2}(A),\ldots,
\alpha^j_{hK^j}(A))$ for $h=1,2$, and $P_0(\bx_1,
\bx_2|A)$ is the conditional marginal likelihood of the combined
sample under a standard OPT as given by
\eqref{eq2}. Equation~\eqref{eq:coOPT} provides a recursive recipe
for computing the marginal likelihood term $P(\bx_1, \bx_2|A)$ for
each potential tree node $A$. (Of course, for this recipe to be of
use, one must also specify the terminal conditions for the
recursion. We will discuss ways to specify such conditions in the next
subsection.)

From \eqref{eq:coOPT} one can tell that the posterior distribution of $(Q_1, Q_2)$ is still a co-OPT distribution through the following reasoning. The first term on the RHS of $\eqref{eq:coOPT}$, $\gamma P_0(\bx_1, \bx_2 |A)$, is the probability (conditional on $A$ being a node reached in the partitioning) of the event
\vspace{-2em}

\begin{gather*}
\text{\{$Q_1$ and $Q_2$ get coupled on $A$, observe $\bx_1(A)$ and $\bx_2(A)$\}.}
\end{gather*}
\vspace{-2em}

The second term, $(1-\gamma)\sum_{j=1}^{M}\lambda_j \frac{D(\bn_1^j+\balp_1^j)D(\bn_2^j+\balp_2^j)}{D(\balp_1^j)D(\balp_2^j)}\prod_{i=1}^{K^j}P(\bx_1,\bx_2|A)$, is the probability of
\vspace{-5em}

\begin{gather*}
\text{\{$Q_1$ and $Q_2$ are not coupled on $A$, observe $\bx_1(A)$ and $\bx_2(A)$\}.}
\end{gather*}
Each summand, $(1-\gamma)\lambda_j
\frac{D(\bn_1^j+\balp_1^j)D(\bn_2^j+\balp_2^j)}{D(\balp_1^j)D(\balp_2^j)}\prod_{i=1}^{K^j}P(\bx_1,\bx_2|A)$,
is the probability of
\vspace{-2em}

\begin{gather*}
\text{\{$Q_1$ and $Q_2$ not coupled on $A$, divide $A$ in the $j$th way, observe $\bx_1(A)$ and $\bx_2(A)$\}.}
\end{gather*}
Finally, given that $C(A)=0$ and $J(A)=j$, the posterior distribution
for $\bthe^j_1$ and $\bthe^j_2$ are Dirichlet$(\bn_1^j+\balp_1^j)$ and
Dirichlet$(\bn_2^j+\balp_2^j)$, respectively. This reasoning, together
with Theorem~3 in \cite{wongandma:2010}, shows that the co-OPT prior
is conjugate, and simple applications of Bayes' Theorem provide the
formulae of the parameter values for the posterior. The results are
summarized in the next theorem.
\begin{thm}
Suppose $\bx_1=(x_{11},x_{12}, \ldots, x_{1n_1})$ and
$\bx_2=(x_{21},x_{22}, \ldots, x_{1n_2})$ are two independent i.i.d{.}
samples from $Q_1$ and $Q_2$. Let $(Q_1, Q_2)$ have a {\it co-OPT}$(\R;
\blam, \balp_1,\balp_2,\gamma; \blam^b, \balp^b, \rho)$ prior that
satisfies the conditions in Theorem~\ref{thm:coOPT}.  Then the
posterior distribution of $(Q_1, Q_2)$ is still a coupling optional
P\'olya tree with the following parameters. 
\begin{enumerate}
\item Coupling probabilities: 
$\quad \gamma(A|\bx_1, \bx_2) = \gamma(A)P_0(\bx_1, \bx_2|A)/P(\bx_1, \bx_2|A).$
\item Partition selection probabilities:
\[ \lambda_j(A |\bx_1, \bx_2) \propto  \lambda_j(A) \frac{D(\bn_1^j+\balp_1^j)D(\bn_2^j+\balp_2^j)}{D(\balp_1^j)D(\balp_2^j)}\prod_{i=1}^{K^j}P(\bx_1,\bx_2|A), \quad j=1,2,\ldots, M(A).\]
\item Probability assignment pseudo-counts:
\[ \alpha^j_{1i}(A |\bx_1, \bx_2) = n_1(A^j_i) + \alpha^j_{1i}(A) \quad \text{and} \quad \alpha^j_{2i}(A |\bx_1, \bx_2) = n_2(A^j_i) + \alpha^j_{2i}(A), \] 
for $j=1,2,\ldots, M(A)$ and $i=1,2,\ldots, K^j(A)$.
\item Base stopping probabilities:
$\quad \rho(A|\bx_1, \bx_2) = \rho(A) u(\bx_1, \bx_2|A)/P_0(\bx_1, \bx_2|A).$
\item Base selection probabilities:
\[ \lambda^b_j(A |\bx_1, \bx_2) \propto  \lambda^b_j(A) \frac{D(\bn_1^j+\bn_2^j+\balp^{bj})}{D(\balp^{bj})}\prod_{i=1}^{K^j}P_0(\bx_1,\bx_2|A), \quad j=1,2,\ldots, M(A).\]
\item Base assignment pseudo-counts:
$\quad \alpha^{bj}_{i}(A |\bx_1, \bx_2) = n_1(A^j_i) + n_2(A^j_i) + \alpha^{bj}_{i}(A),$\\
for $j=1,2,\ldots, M(A)$ and $i=1,2,\ldots, K^j(A)$.
\end{enumerate}
\end{thm}
\noindent Two remarks: (1) All of the posterior parameter values can be
computed {\em exactly} using the above formulae, without the need of any
Monte Carlo procedure. (2) The posterior coupling parameters contain
information about the difference
between the two underlying distributions $Q_1$ and $Q_2$, while the
posterior base parameters contain information regarding the
underlying structure of the two measures. This naturally
suggests that if one is only interested in two sample comparison, one should
only need the posterior distribution of the coupling variables, and not
those of the base variables. This will become clear in Sections~5 and 6
where we give several numerical examples.

\subsection{Terminal conditions}
As mentioned earlier, we need to specify the terminal conditions for
the recursion used to compute $P(\bx_1, \bx_2 |A)$. 
Depending on the nature of $\om$ and the prior specification, the
recursion formula \eqref{eq:coOPT} can terminate in several ways as
demonstrated in the following two examples.

\begin{exam}[$2^p$ contingency table]
Let $\om=\{1,2\}\times\{1,2\}\times\dots\times\{1,2\}$. For
any rectangle $A$ in the table---a set of the form $A_1 \times A_2
\times \dots A_p$ with $A_1$, 
$A_2$, \ldots, $A_p$ being non-empty subsets of $\{1,2\}$---let
$k_1, k_2, \ldots, k_{M(A)}$ be the ``intact'' dimensions of $A$, that
is $A_{k_j}=\{1,2\}$ for $j=1,2,\ldots, M(A)$. Let $\R$ be the diadic
splitting rule that allows $A$ to be divided into two halves on each intact
dimension $j$. In our earlier notation,
$ \R(A)=\left\{\{A^j_1, A^j_2\}_{j=1}^{M(A)}\right\},$
where $A^j_1 = A_1 \times A_2 \times \dots \times A_{k_j-1} \times
\{1\} \times A_{k_j+1} \times \dots \times 
A_p$ and $A^j_2 = A_1 \times A_2 \times \dots \times A_{k_j-1} \times
\{2\} \times A_{k_j+1} \times \dots \times 
A_p$. Suppose two i.i.d.\! samples  $\bx_1$ and $\bx_2$ are observed.
Assume that $(Q_1, Q_2)$ has
a co-OPT prior with the following prior parameter values for each
rectangle $A$:  
$\lambda_{j}(A)=\lambda^b_j(A)=\tfrac1{M(A)}$,
$\alpha_{i}^j(A)=\alpha_{i}^{bj}(A)\equiv \tfrac12$ for $i=1,2$ and
$j=1,2,\ldots, M(A)$, and finally  $\gamma(A)
\equiv \gamma_0$, $\rho(A)\equiv \rho_0$, where $\gamma_0$ and
$\rho_0$ are constants in $(0,1)$.  

In this example, there are three types of terminal nodes for
$P_0(\bx_1, \bx_2 |A)$ and they are given in
Example~3 of \cite{wongandma:2010}.
By a similar reasoning, there are also three types of terminal nodes
for $P(\bx_1,\bx_2|A)$.
\begin{enumerate}
\item If $A$ contains no data point from either sample,
  $P(\bx_1, \bx_2|A)=1$.
\item If $A$ is a single table cell containing any number
  observations, $P(\bx_1, \bx_2 |A)=1$. 
\item $A$ contains a single observation (from either sample). In this
  case, $P(\bx_1, \bx_2 |A)=2^{-M(A)}=1/\mu(A)$. To see this, first we
  let $t_{M(A)} = P(\bx_1,\bx_2|A).$
  By Example~3 in \cite{wongandma:2010}, we have
  $P_0(\bx_1,\bx_2|A) = 2^{-M(A)}.$
  Hence we have
\begin{align*}
t_{M(A)}&=\gamma_0 2^{-M(A)}+(1-\gamma_0)
\left(\dfrac1{M(A)}\sum_{j=1}^M\frac{B\left(\tfrac32,\tfrac12\right)}{B\left(\tfrac12,\tfrac12\right)}\right)\cdot t_{M(A)-1}\\
   &=\gamma_0 2^{-M(A)}+(1-\gamma_0)\tfrac12 t_{M(A)-1}\\
   &=\gamma_0 2^{-M(A)}\,\frac{\left(1-(1-\gamma_0)^{M(A)}\right)}{1-(1-\gamma_0)}+\left(\frac{1-\gamma_0}{2}\right)^{M(A)}\\ 
   &=2^{-M(A)}=1/\mu(A).
\end{align*}
\end{enumerate}
\label{ex:2toK}
\end{exam}
\vspace{-2em}

\begin{exam}[Rectangle in $\real^p$]
Let $\om=I_1 \times I_2 \times \ldots \times I_p$ be a bounded
rectangle in $\real^p$. Let $\R$ be the 
diadic partition rule such that for any rectangle $A$ of the form
$A_1 \times A_2
\times \dots A_p$ with $A_1$, 
$A_2$, \ldots, $A_p$ being non-empty subintervals of
$I_1,I_2,\ldots,I_p$ respectively, $A$ can be divided in half in any of the
$p$ dimensions. Again, let $\bx_1$ and $\bx_2$ be the two samples, and
let $(Q_1,Q_2)$ have a co-OPT prior with the following parameters:
$\lambda_{j}(A)=\lambda^b_j(A)\equiv\tfrac1{p}$,
$\alpha_{i}^j(A)=\alpha_i^{bj}(A) \equiv \tfrac12$, $\gamma(A)\equiv \gamma_0$
and $\rho(A)\equiv \rho_0$, for all $A$, $i=1,2$, and $j=1,2,\ldots, M(A)$.

In this case there are two types of terminal nodes
for $P(\bx_1, \bx_2 |A)$.
\begin{enumerate}
 \item $A$ contains no observations. In this case,
  $P(\bx_1, \bx_2|A)=1$.
\item $A$ contains a single observation (from either sample). Then
  $P(\bx_1,\bx_2|A)=1/\mu(A)$. We skip the derivation
  of this as it is similar to that used for Case~3 in Example~\ref{ex:2toK}.
\end{enumerate}
Note that in this example we have implicitly assumed that no
observations, from either sample,
can be identical. With the assumption that $Q_1$ and $Q_2$ are absolutely
  continuous w.r.t{.} the Lebesgue measure, the probability for any observations to be identical
  is 0. However this situation can occur in real data due to rounding. This
  possibility can be dealt with in our following discussion on
  technical termination of the recursion.
\label{ex:RtoK}
\end{exam}
Other than the ``theoretical'' terminal nodes given in the previous
two examples, in real applications it is often desirable to set a
technical lower limit on the size of the nodes to be computed in
order to save computation.  For instance, in the
$\real^p$ example, one can impose that all nodes smaller than 1/1000 of the space $\om$ be
stopped and coupled. That is to let $\gamma(A) = \rho(A) =  1$ by design for all small enough
$A$. The appropriate cutoff threshold of the node size depends on
the nature of the data, but typically there is a wide range of values
that work well. For most problems such a
technical constraint should hardly have any impact on the posterior
parameter values for large nodes. It is worth emphasizing
that for real-valued data, which are almost always discretized (due to rounding), such a constraint actually
becomes useful also in preventing numerical anomalies. In such cases, a general
rule of thumb is that one should always adopt a cutoff size larger than the
rounding unit relative to the length of the corresponding boundary of
the space.  

\section{Numerical examples on two sample comparison}
We next provide three numerical examples,
Examples~\ref{ex:R} through \ref{ex:2toK_ROC}, to demonstrate
inference on the two sample problem using the co-OPT prior. In
these examples, the posterior coupling probability of $\om$ serves as
a statistic for testing whether the two samples have come from the
same distribution, which we will refer to as the co-OPT
statistic. In each example we compare our
method to one or more other existing approaches, and in Example~\ref{ex:2toK_ROC}
 we show how
the posterior values of the coupling variables can be used to
learn the underlying structure of the discrepancy between the two
samples. 

For all these examples, we set the prior parameter values in the
fashion of Examples~\ref{ex:2toK} and \ref{ex:RtoK}, with $\gamma_0=\rho_0=0.5$. In Examples~\ref{ex:R} and \ref{ex:R2}, whenever the underlying
distributions have unbounded support, we simply use the range of
the data points in each dimension to define the rectangle $\om$. (As a referee pointed out, an alternative to using this data dependent support is to transform each unbounded dimension through a measurable map such as a cumulative distribution function. The choice of such maps will influence the underlying inference. Although we do not investigate this relation in the current work, it is certainly interesting and deserve further studies in future works.) Also, in these three examples we use
1/1000 as the size cutoff for ``technical'' terminal nodes as
discussed in the previous section.

\begin{exam}[Two sample problem in $\real$]
We simulate the control and case samples under the following three scenarios.
\begin{enumerate}
\item Locational shift: Sample~1 $\sim$ Beta(4,6) and Sample~2 $\sim$
  0.2 + Beta(4,6) with sample sizes $n_1=n_2=20$.
\item Local structure: Sample~1 $\sim$ Uniform[0,1] and Sample~2
  $\sim$ 0.5 Beta(20,10) + 0.5 Beta(10,20) with $n_1=n_2=30$.
\item Dispersion difference: Sample~1 $\sim$ N(0,1) and Sample~2
  $\sim$ N(0,4) with $n_1=n_2=40$.
\end{enumerate}
We place a co-OPT prior on $(Q_1, Q_2)$ 
as described in Example~\ref{ex:RtoK}. (Because here there is
only one dimension, there is no choice of ways to split.) We compare the ROC curves of four
different statistics for testing the null hypothesis that the two samples
have come from the same distribution---namely the Kolmogorov-Smirnov
(K-S) statistic \cite[pp. 392--394]{chakravartistat}, Cramer-von-Mises (CvM) statistic \cite{anderson:1962},
Cramer-test statistic \cite{baringhuas:2004}, and
our co-OPT statistic. The results are presented in the middle column of \ref{fig:ex_1d}. In addition, we also investigate the power of each statistic at the 5\% level under various sample sizes, ranging from 10 data points per sample to 60 per sample. (See the right column of \ref{fig:ex_1d}.) Our co-OPT statistic behaves worse than the other three tests under
the first scenario when there is a simple locational shift, 
better than the other tests for the second scenario, slightly
worse than the Cramer test but
better than the K-S and CvM tests under the last scenario.
\label{ex:R}
\end{exam}
\begin{figure}[tb]
\begin{center}
    \leavevmode 
    \includegraphics[width=16cm]{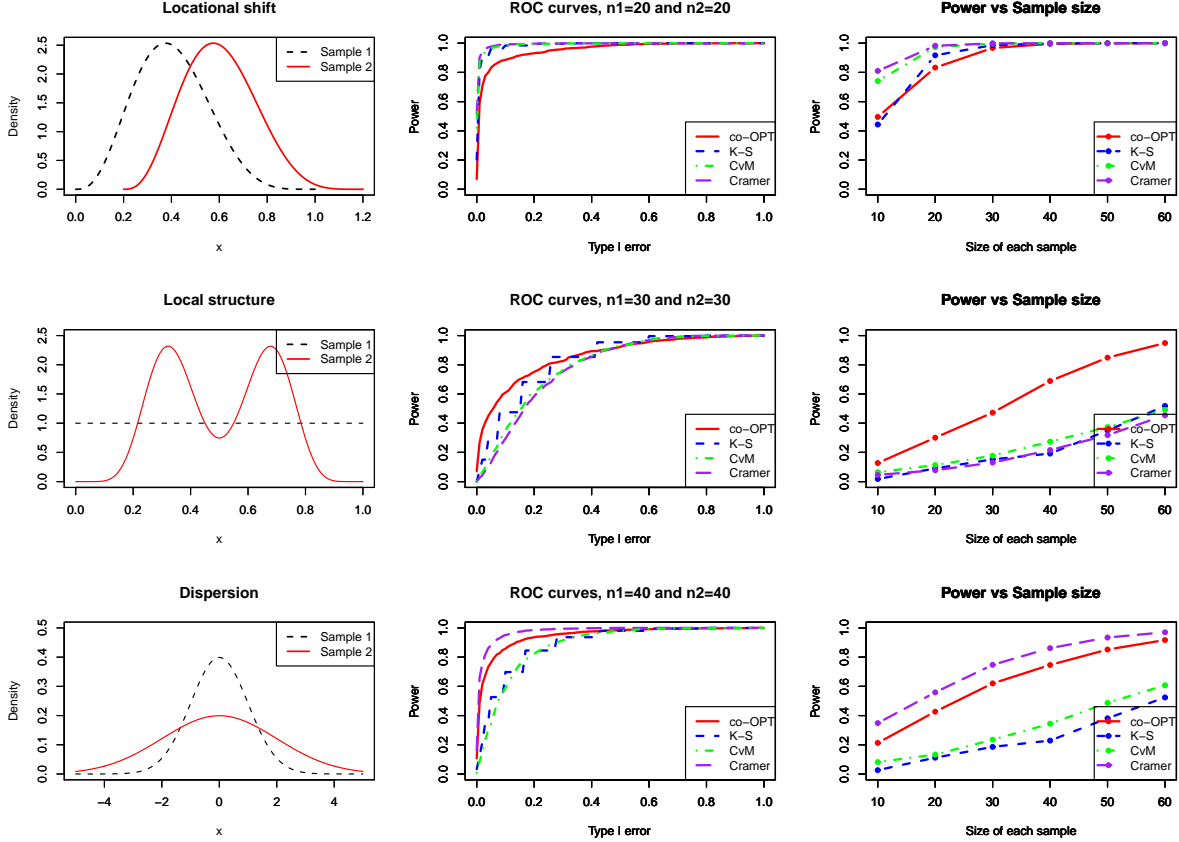}
     \caption{Two simulated samples on $\real$ under three scenarios
       (rows) given in Example~\ref{ex:R}. 
       Left panel: Density functions for the
       two samples. Middle panel: ROC curves for four
       test statistics. Right panel: Power vs. sample size---power (at the 5\% level) is estimated from simulation under equal sample size (horizontal axis) of the case group and the control group. }
  \label{fig:ex_1d}
\end{center}
\end{figure}

\begin{exam}[Two sample problem in $\real^2$]
We simulate two samples under four scenarios.
\begin{enumerate}
\item Locational shift ($n_0=n_1=50$): 

  Sample~1 $\sim BN\Biggl(\begin{pmatrix}
  1 \\ 0 \end{pmatrix},\begin{pmatrix} 2^2 & 0 \\ 0 &
  2^2\\\end{pmatrix}\Biggr)$ and Sample~2 $\sim BN\Biggl(\begin{pmatrix}
  0 \\ 1 \end{pmatrix},\begin{pmatrix} 2^2 & 0 \\ 0 &
  2^2\\\end{pmatrix}\Biggr)$.

\item Subset shift ($n_0=n_1=100$):

  Sample~1 $\sim BN\Biggl(\begin{pmatrix}
  0 \\ 0 \end{pmatrix},\begin{pmatrix} 0.3^2 & 0 \\ 0 &
  0.3^2\\\end{pmatrix}\Biggr)$ and \\
  Sample~2 $\sim 0.8\times BN\Biggl(\begin{pmatrix}
  0 \\ 0 \end{pmatrix},\begin{pmatrix} 0.3^2 & 0 \\ 0 &
  0.3^2\\\end{pmatrix}\Biggr) + 0.2\times BN\Biggl(\begin{pmatrix}
  0.5 \\ 0.5 \end{pmatrix},\begin{pmatrix} 0.3^2 & 0 \\ 0 &
  0.3^2\\\end{pmatrix}\Biggr)$.

\item Dispersion difference ($n_0=n_1=50$): 

  Sample~1 $\sim BN\Biggl(\begin{pmatrix}
  0 \\ 0 \end{pmatrix},\begin{pmatrix} 1 & 0 \\ 0 &
  1\\\end{pmatrix}\Biggr)$ and Sample~2 $\sim BN\Biggl(\begin{pmatrix}
  0 \\ 0 \end{pmatrix},\begin{pmatrix} 0.5^2 & 0 \\ 0 &
  0.5^2\\\end{pmatrix}\Biggr)$.
\item Local structure ($n_0=n_1=50$): 

Sample~1 $\sim   BN\Biggl(\begin{pmatrix}
  0 \\ 0 \end{pmatrix},\begin{pmatrix} 1 & 0.5^2 \\ 0.5^2 &
  1\\\end{pmatrix}\Biggr)$, and

Sample~2 $\sim 0.5\times BN\Biggl(\begin{pmatrix}
  0.5 \\ 0.5 \end{pmatrix},\begin{pmatrix} 0.4^2 & 0 \\ 0 &
  0.4^2\\\end{pmatrix}\Biggr) + 0.5\times BN\Biggl(\begin{pmatrix}
  -0.5 \\ -0.5 \end{pmatrix},\begin{pmatrix} 0.4^2 & 0 \\ 0 &
  0.4^2\\\end{pmatrix}\Biggr)$.
\end{enumerate}
We compare four statistics that measure the similarity between two
distributions 
---(1) the co-OPT statistic, (2) the Cramer test statistic 
 \cite{baringhuas:2004}, (3) the log Bayes factor under P\'olya tree (PT)
 priors in \cite{holmes:2009},
 and (4) the posterior mean of the ``similarity parameter'' $\epsilon$
 given in the dependent Dirichlet Process mixture (DPM) prior proposed in
 \cite{muller2004}. The ROC curves are presented in 
\ref{fig:ex_2d}. Again, the co-OPT performs relatively poorly for a simple
(global) locational shift, but performs resonably well under the other
three scenarios. The PT method does not allow adpative partitioning of the
space, and that appears to have cost a lot of power. On the other hand, the DPM
method performs well under all but the subset shift scenario. Because
the similarity parameter $\epsilon$ captures the proportion of
``commonness'' between two distributions \cite{muller2004}, it does
not capture well differences pertaining to only a small portion of the
probability mass. Details about the prior
specifications for the PT and DPM methods can be found in the
supplementary material. Note that there may be alternative
specifications that will lead to better performance of
these methods for the current example.
\label{ex:R2}
\end{exam}

\begin{figure}[htb]
\begin{center}
    \leavevmode 
    \includegraphics[width=15cm]{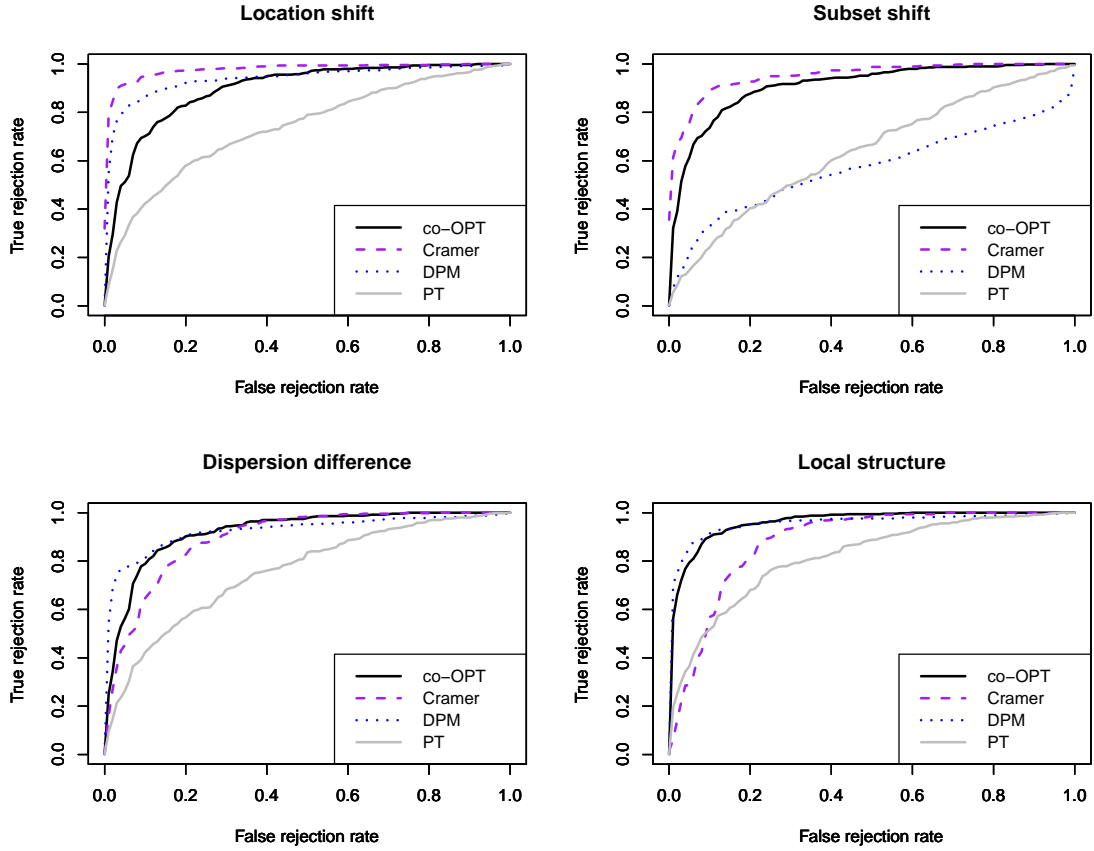}
     \caption{ROC curves for two samples on $\real^2$ under
       the four scenarios given in Example~\ref{ex:R2}.}
  \label{fig:ex_2d}
\end{center}
\end{figure}

Our next example deals with retrospectively sampled data on a
high-dimensional contingency table. In this example, we not only
demonstrate 
the power of our method to test for two sample difference, but also
show that the posterior co-OPT distribution can help learn the underlying
structure of the difference. 

\begin{exam}[Retrospectively sampled data on a $2^{15}$ contingency table]
Suppose there are 15 binary predictors $X_{1}, X_{2}, \ldots, X_{15}$,
and there is a binary response variable $Y$,
e.g. disease status, whose
distribution is
\[
Y \sim  \left\{ \begin{array}{ll}
Bernoulli(0.3) & \mbox{if $X_3 = 1$ and $X_7 = 1$} \\
Bernoulli(0.3) & \mbox{if $X_7 = 0$ and $X_{10} = 0$} \\
Bernoulli(0.1) & \mbox{otherwise.}
\end{array}
\right.
\]
We simulate populations for joint observations of $X_i$'s and $Y$ of
size 200,000 under two scenarios 
\begin{enumerate}
\item $X_1, X_2, \ldots, X_{15}$ $\sim_{i.i.d.}$ Bernoulli(0.5)
\item $X_1, X_2, \ldots X_{8}$ as a Markov Chain with $X_1 \sim$
  Bernoulli(0.5), and $P(X_t = X_{t-1}|X_{t-1})=0.7$, while $X_9, X_{10},
  \ldots X_{15}$ $\sim_{i.i.d}$ Bernoulli(0.5) and are independent of
  $X_1, \ldots, X_8$.
\end{enumerate}
For each scenario, we retrospectively sample controls ($Y$=0) and
cases ($Y$=1). Our interest is in (1) the power of our method in
detecting the difference in the joint distribution of the predictor variables
between the two samples, and (2) whether the method can recover the
``interactive'' structure among the three predictors $X_3$, $X_7$ and $X_{10}$. 

We place two different priors on $(Q_1,Q_2)$ and compare their
performance. The first is our co-OPT distribution with prior
parameters being specified as in Example~\ref{ex:2toK}. The second is a dependent
Dirithlet prior inspired by \cite{muller2004}.
Under this setup we write $Q_1$ and $Q_2$ as mixtures,
$Q_1 = \epsilon H_0 + (1-\epsilon) H_1$ and 
$Q_2 = \epsilon H_0 + (1-\epsilon) H_2$,
where $H_0$ represents the common part of $Q_1$ and $Q_2$ while $H_1$
and $H_2$ the idiosyncratic portion. Under the prior, 
$H_0$, $H_1$ and $H_2$ $\sim_{i.i.d} Dirichlet({\bm \alpha_H})$
and $\epsilon \sim Beta(a_{\epsilon},b_{\epsilon})$. For the hyperparameters,
we chose ${\bm \alpha_H}=(0.5,0.5,\ldots,0.5)$---that is, each cell in
the support of the Dirichlet receives 0.5 prior pseuodocount, and
$a_{\epsilon}=b_{\epsilon}=3$.  We found
that due to the sparsity of the table, restricting the prior to have a
support over only the observed table cells rather than the entire table
drastically improves the power. Therefore, this is what we do
here. (More details
about the prior specification and how MCMC sampling is used to draw
posterior samples for this prior can be found in the supplementary materials.) 

Because $\epsilon$ can be thought of as
a measure of how similar $Q_1$ and $Q_2$ are, the mean of its
posterior distribution can serve as a statistic (which we shall from now on refer to as the
$\epsilon$-statistic) for testing the difference between the
two. The ROC curves of the $\epsilon$-statistic, and that of our co-OPT
statistic, for the two scenarios and different sample sizes are given in the left and middle columns of
\ref{fig:ex1_2toK}. For comparison, the right column of the figure gives the ROC curve for another statistic measuring two sample difference, namely the empirical $L_2$ distance between the two contingency tables corresponding to the cases and the controls. Note that in this example to achieve comparable performance the
the $\epsilon$-statistic and the $L_2$ distance both require samples sizes 10 times as large as
those for the co-OPT! This performance advantage of the co-OPT in this setting
is probably due to (1) the adaptive partitioning feature and (2) the
coupling feature, both of which help mitigate the difficulties caused by the sparsity of the table
counts. Also interesting is the impact of the
correlation among predictors on the power. For the co-OPT, the
correlation structure in Scenario~2 makes it harder to find a good
partition of the space and therefore reduces power. On the other hand,
the performance of the $\epsilon$-statistic, as well as that of the $L_2$ distance, is actually better for
Scenario~2, as the correlation structure turns a marginal association
(marginal w.r.t. the subspace of $X_3$, $X_7$ and $X_{10}$) into a
joint one involving $X_1$ through $X_{10}$.

While the $\epsilon$-statistic and the $L_2$ distance can only serve for
detecting the difference, the posterior co-OPT can also
capture the underlying structure of the difference. We find that with about 500 data points in
each sample for Scenario~1 and about 3500 data points in each sample for
Scenario~2, the underlying structure can be accurately recovered using
the hierarchical {\it maximum a posteriori} (hMAP)
tree topology, which is a top-down stepwise posterior maximum likelihood
tree. (The construction of the hMAP tree as 
well as the motivation to choose it over the MAP tree is
discussed in detail in Section 4.2 of \cite{wongandma:2010}.) As one would expect, the correlation between the
predictor variables makes it much harder to recover the exact interactive
relation. A typical hMAP tree
structure for the simulated populations with these sample sizes is given in
\ref{fig:ex_2toK_hmap}. We note that in general a sample of partition trees from the posterior distribution of
the tree structure can be more informative
than the hMAP tree, especially when the sample sizes are not
large enough. We use the hMAP here as a demonstration for its
ease of visualization.
\begin{figure}[!htbp]
\begin{center}
    \leavevmode 
    \includegraphics[width=13cm]{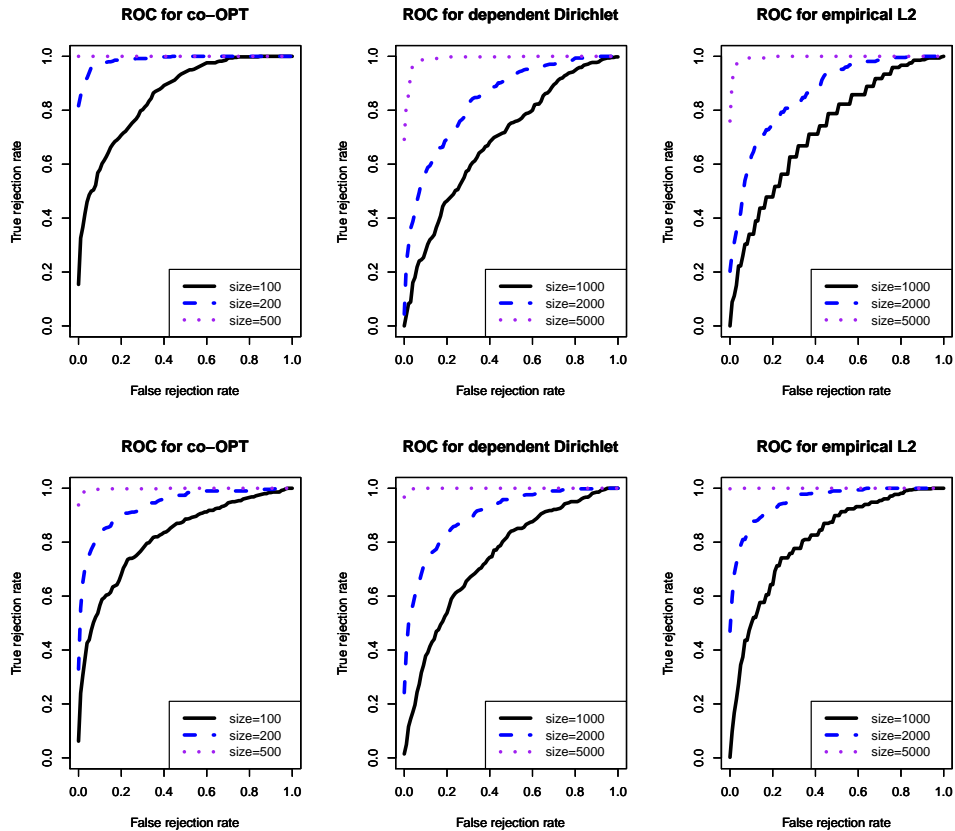}
     \caption{ROC curves of the co-OPT (left), the $\epsilon$ (middle), and the empirical $L_2$ (right) statistics
       for the two scenarios given in Example~\ref{ex:2toK_ROC} (first row for Scenario 1 and second row for Scenario 2). The sample sizes for
  $\epsilon$ and $L_2$ are 10 times as large as those for the co-OPT.}
  \label{fig:ex1_2toK}
\end{center}
\end{figure}

\begin{figure}[!htbp]
\begin{center}
    \leavevmode 
    \includegraphics[width=11cm]{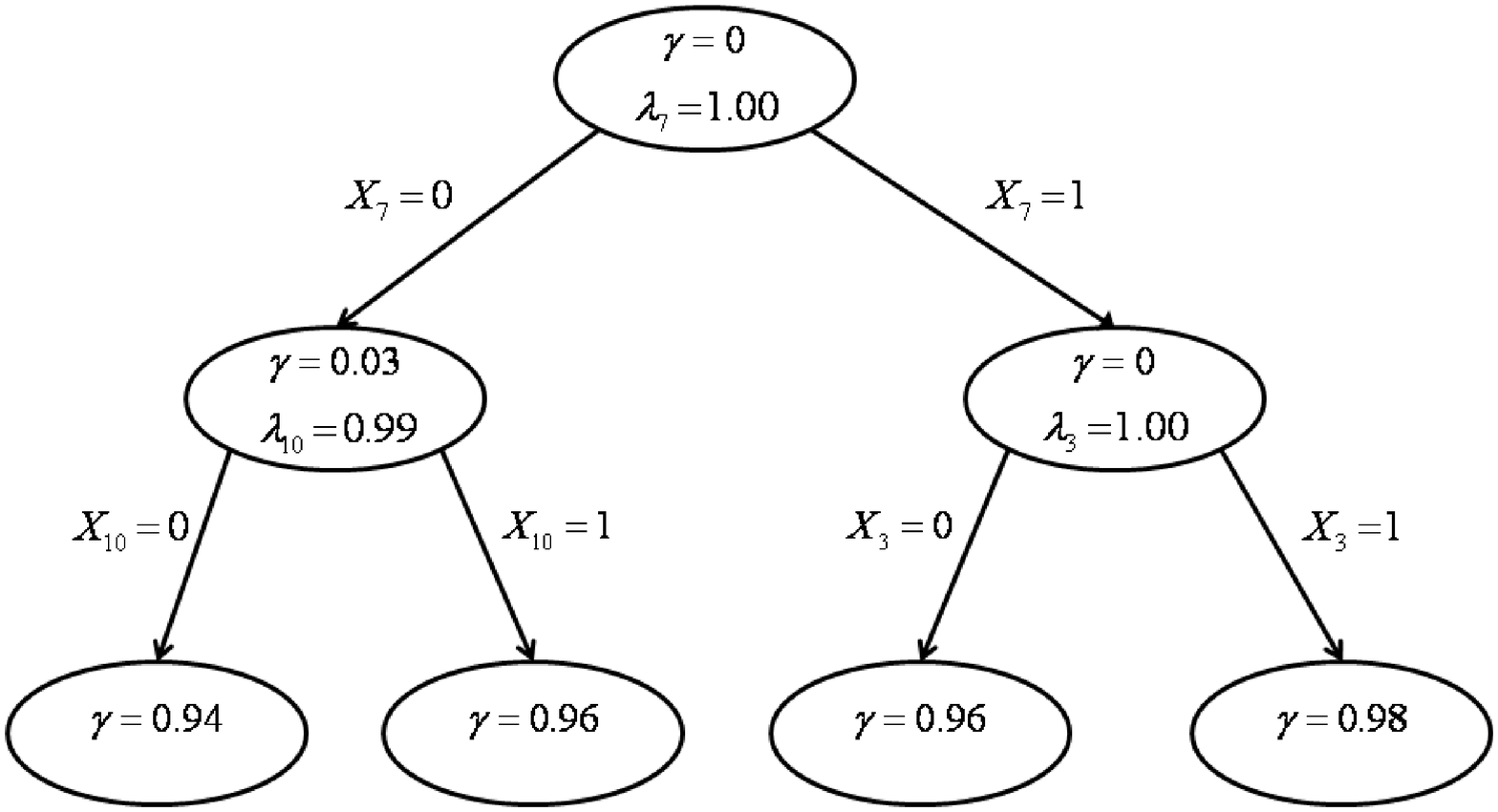}
    \caption{A typical hMAP coupling tree that recovers the underlying interactive
      structure.}
  \label{fig:ex_2toK_hmap}
\end{center}
\end{figure}
\label{ex:2toK_ROC}
\end{exam}

\section{Inference on distributional distances between two samples}
In some situations, one may be interested in a distance measure for the
two sample distributions. For example, if we let $d(Q_1,Q_2)$ denote the distance
between the two sample distributions under some metric $d$, one may
want to compute quantities such as $P(d(Q_1,Q_2) > T | \bx_1, \bx_2)$ where $T$ is some constant.
This can be achieved if one knows the posterior distribution of
$d(Q_1,Q_2)$ or can sample from it. We next show that if $(Q_1,Q_2)$ arises from
a co-OPT distribution, then for some common metrics, in particular $L_1$
and Hellinger distances, it is very convenient to sample from the
distribution of $d(Q_1,Q_2)$. 

As before, let $Q_1$ and $Q_2$ (with densities $q_1$ and
$q_2$ respectively) be the two distributions of interest. Suppose
$(Q_1,Q_2)$ have a co-OPT distribution, and so they can
be thought of as being generated from the
random-partitioning-and-assignment procedure introduced in the
previous section through the drawing
of the variables $C$, $J$, $\bthe_1$, $\bthe_2$, $C^b$, $J^b$ and $\bthe^b$. 
 Then we have the following result.
\begin{prop}
Suppose $(Q_1,Q_2)$ has a co-OPT distribution satisfying the
conditions given in Theorem~\ref{thm:coOPT}.
Let $\A(C,J)$ denote the (random) collection of all nodes on which $Q_1$ and $Q_2$ first
couple. (The notation indicates that it depends on the coupling
variables $C$ and $J$.) Also, let $d_{L_1}$ be the $L_1$ distance, and
$d_{H^2}$ the squared Hellinger distance. (That is,
$d_{L_1}(f,g)=\int|f-g|$ and $d_{H^2}(f,g)=\int
(\sqrt{f}-\sqrt{g})^2$.) Then 
\begin{align*}
d_{L_1}(Q_1, Q_2) &= \sum_{A\in \A(C,J)} |Q_1(A)-Q_2(A)|\\
d_{H^2}(Q_1,Q_2) &= \sum_{A \in \A(C,J)} (\sqrt{Q_1(A)}-\sqrt{Q_2(A)})^2.
\end{align*}
\label{prop:dist}
\end{prop}
\vspace{-3em}

\begin{proof}
See supplementary materials.
\end{proof}
This proposition provides a recipe for drawing samples from the
distributions of $d_{L_1}(Q_1,Q_2)$ and $d_{H^2}(Q_1,Q_2)$. One can first draw the coupling variables
$C$, $J$, $\bthe^1$ and $\bthe^2$. Then use $C$ and $J$ to find the
collection of nodes $\A(C,J)$, and use $\bthe^1$ and $\bthe^2$ to
compute, for each $A \in \A(C,J)$, the corresponding measures $Q_1(A)$
and $Q_2(A)$. Finally, one draw of $d_{L_1}$ (or $d_{H^2}$) can be
computed by summing
$|Q_1(A)-Q_2(A)|$ (or $(\sqrt{Q_1(A)}-\sqrt{Q_2(A)})^2$) over all
nodes in $\A(C,J)$. 

A particularly desirable feature of this procedure for sampling $L_1$ and
Hellinger distances is that one does not need to draw samples for the two random
distributions $Q_1$ and $Q_2$ to get their distances. In fact, one
only needs to draw the coupling variables, which characterize the {\em
  difference} between the two distributions, without having to draw
the base variables, which characterize the fine structure of the two
densities. Again, in multi-dimensional settings where estimating
densities is difficult, such a procedure can produce much less
variable samples for the distances.

We close this section with two more numerical examples, one in $\real$ and one in
$\real^2$. In the second of these, again we use the observed range of
the data in each dimension to define the space $\om$. Also, we use
1/10000 as the size cutoff for technical termination.
\begin{exam}[Two beta distributions] 
We simulate two samples from Beta(2,5) and Beta(20,15) under three
sets of sample sizes
$n_1=n_2=$10, 100 and 1000. We place a co-OPT prior on the two distributions
with the diadic partition rule and the symmetric parameter values as
specified in Example~\ref{ex:RtoK} with $\rho_0=\gamma_0=0.5$, and
compute the corresponding posterior co-OPT. Then we
draw 1000 samples for each of $d_{L_1}(Q_1,Q_2)$ and $d_{H^2}(Q_1,Q_2)$ from
their posterior distributions. The histograms of these samples are
plotted in \ref{fig:dist_1D}, where the vertical lines
indicate the actual $L_1$ and squared Hellinger distances between the two distributions.
\begin{figure}[!htb]
\begin{center}
    \leavevmode 
    \includegraphics[width=14cm]{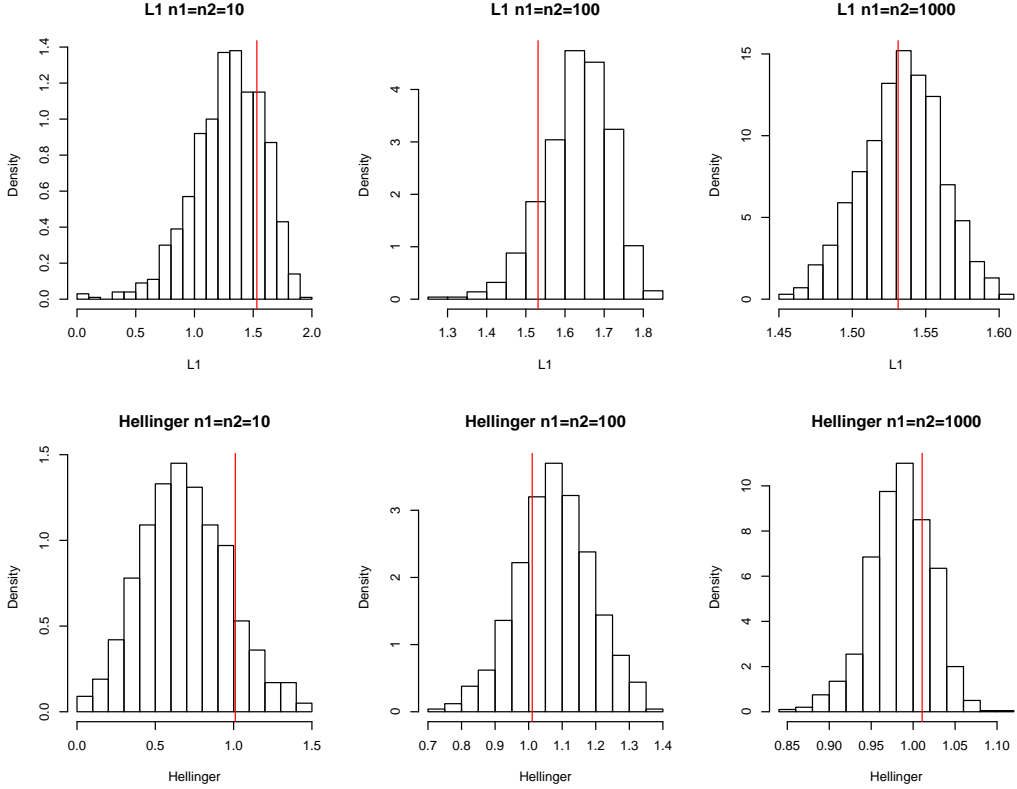}
    \caption{Histograms for posterior samples of $L_1$ and squared Hellinger
      distances for two samples from Beta(2,5) and Beta(20,15). The
      vertical lines indicate the actual $L_1$ and squared Hellinger
      distance between these two distributions.}
  \label{fig:dist_1D}
\end{center}
\end{figure}
\label{ex:dist_1D}
\end{exam}

\begin{exam}[Bivariate normal and mixture of bivariate normal]
We repeat the same thing as in the previous example except now we
simulate the two samples from the following distributions in
$\real^2$.

Sample~1 $\sim   BN\Biggl(\begin{pmatrix}
  0 \\ 0 \end{pmatrix},\begin{pmatrix} 4 & 0 \\ 0 &
  4 \\\end{pmatrix}\Biggr)$, and

Sample~2 $\sim 0.5\times BN\Biggl(\begin{pmatrix}
  1 \\ 1 \end{pmatrix},\begin{pmatrix} 1 & 0 \\ 0 &
  1 \\\end{pmatrix}\Biggr) + 0.5\times BN\Biggl(\begin{pmatrix}
  -1 \\ -1 \end{pmatrix},\begin{pmatrix} 1 & 0 \\ 0 &
  1 \\\end{pmatrix}\Biggr)$.\\

Again we draw 1000 posterior samples for $d_{L_1}(Q_1,Q_2)$ and for
$d_{H^2}(Q_1,Q_2)$ under each set of sample sizes. The histograms of
these samples are 
plotted in \ref{fig:dist_2D}, where the vertical lines again
indicate the actual $L_1$ and squared Hellinger distances between the two distributions.
\begin{figure}[!htb]
\begin{center}
    \leavevmode 
    \includegraphics[width=14cm]{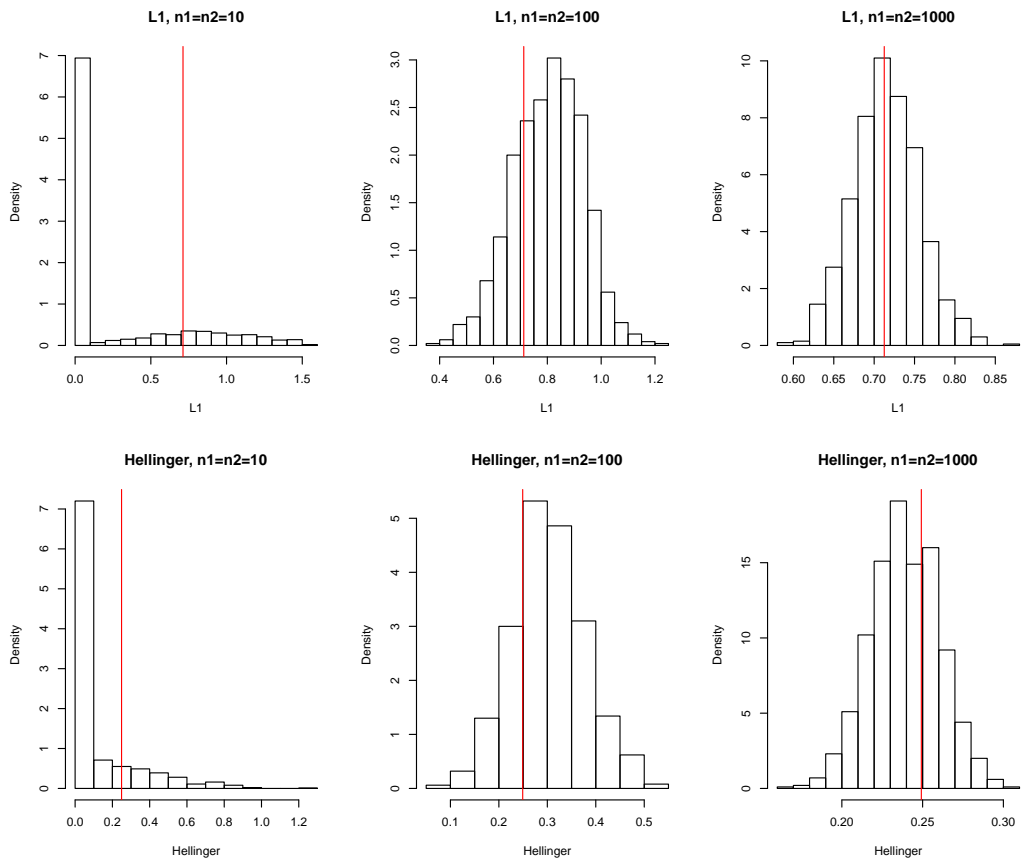}
    \caption{Histograms for posterior samples of $L_1$ and squared Hellinger
      distances for Example~\ref{ex:dist_2D}. The vertical lines
      indicate the actual $L_1$ and squared Hellinger distances for the two
      underlying distributions.}
  \label{fig:dist_2D}
\end{center}
\end{figure}
\label{ex:dist_2D}
\end{exam}

\section{Concluding remarks}
In this work we have introduced the coupling optional P\'olya tree prior for Bayesian nonparametric analysis 
on the two sample problem. This prior jointly generates two
random probability distributions that can ``couple''
on subsets of the sample space.  We have demonstrated that this
construction allows both the testing and the learning of the
distributional difference between the two samples. One can easily extend this
prior to allow the joint generation of more than two samples. For
example, if four samples are involved, then one can draw four, instead
of two,  independent Dirichlet vectors for probability assignment on each
uncoupled node. 

One interesting feature of the co-OPT prior (as well as the original
OPT prior) is that the corresponding
posterior can be computed ``exactly'' using the recursive
formulation given in \eqref{eq:coOPT_lik} without resorting to Markov Chain Monte
Carlo sampling. However, such ``exact inference''
based on recursions is still computationally
intensive, especially in high-dimensional problems. Efficient implementation is a
necessity for this method to be feasible for any non-trivial
problems. 
However, even with the most efficient
implementation, the exponential nature of the method dictates that 
approximation techniques such as $k$-step look-ahead as well as
large-scale parallelization are needed for very high dimensional
problems, such as those on a contingency table with 100 dimensions. Current work is undergoing in
this direction.     

\subsection*{Acknowledgment}
The authors want to thank Art Owen, David Siegmund, Hua Tang,
Robert Tibshirani, four referees, and the editors for very helpful
comments. LM is
supported by a Larry Yung Stanford Interdisciplinary Graduate Fellowship. WHW
is supported in part by NSF grant DMS-0906044. Much of the computation
is done on systems supported by NSF SCREMS Grant DMS-0821823 and NSF award CNS-0619926.

\subsection*{Supplemental materials}
\begin{description}
\item[Appendices] A.1 includes all the proofs. A.2
  gives the details about the prior specifications for the comparison made in
  Example~\ref{ex:R2}. A.3 gives the details about the specification
  of the dependent Dirichlet prior as well as the Gibbs sampler used
  for drawing posterior samples of $\epsilon$.
\end{description}

\section*{Appendix A.1. Proofs}
\begin{proof}[Proof of Theorem~1]
Consider the RPAA procedure described in Section~2 with the
uniform base distribution $u$ replaced
by $q_0$. So under this new procedure of generating a random measure
$Q$, whenever a region $A$ gets stopped, the conditional
distribution of $Q$ within $A$ is set to be $Q_0(\cdot |A)$. Let $Q^{(k)}$ be
the corresponding random distribution that is forced
to stop after $k$ levels of nested partitioning. In other words, for
all non-stopped nodes $A$ reached after $k$ levels of nested partitioning, we
stop dividing $A$ regardless of the
stopping variable $S(A)$ and force a conditional
distribution $Q_0(\cdot |A)$ on it to obtain $Q^{(k)}$. (For more
detail see the proof of Theorem~1 in \cite{wongandma:2010}.)

We first show that if
$\alpha^j_i(A)/\sum_{h=1}^{K^j(A)}\alpha^j_h(A)=Q_0(A^j_i)/Q_0(A)$,
then $EQ^{(k)}(B)=Q_0(B)$ for all $k$. For
$k \geq 0$, let $\J^{(k)}$ be the collection of all partition random variables
$S$ and $J$ drawn in the first $k$ levels of partitioning, and let
$\A(\J^{(k)})$ be the collection of all leaf nodes after $k$ levels of
random partitioning---those are the nodes that
are either just reached in the $k$th step or are reached earlier but stopped. We prove by
induction that $E\left(Q^{(k)}(B) | \J^{(k)}\right)=Q_0(B)$.   For
$k=0$, $\J^{(k)}=\emptyset$, $\A(\J^{(k)})=\{\om\}$, and $Q^{(0)}=Q_0$
and so $E\left(Q^{(k)}(B)|\J^{(k)}\right)=Q_0(B)$ holds trivially. Now
for $k \geq 1$, suppose this holds true for $1,2,\ldots,k-1$. By construction, 
\[Q^{(k)}(B)=\sum_{A\in \A(\J^{(k)})} Q^{(k)}(A)\frac{Q_0(B\cap
  A)}{Q_0(A)}. \] 
Let $A^p \in \A(\J^{(k-1)})$ be the parent node of $A$, that is, the
node whose division gives rise to $A$. Then by
the condition that
$\alpha^j_i(A)/\sum_{h=1}^{K^j(A)}\alpha^j_h(A)=Q_0(A^j_i)/Q_0(A)$, we
have \[E(Q^{(k)}(A)/Q^{(k)}(A^p)|\J^{(k)})=Q_0(A)/Q_0(A^p),\] 
and so
\begin{align*}
E\left(Q^{(k)}(B)|\J^{(k)}\right)&=E\left(\sum_{A\in \A(\J^{(k)})} Q^{(k)}(A)\frac{Q_0(B\cap
    A)}{Q_0(A)} \Big | \J^{(k)}\right)\\
&=\sum_{A\in \A(\J^{(k)})} \frac{Q_0(B\cap
    A)}{Q_0(A)} E\Biggl( Q^{(k)}(A) \Big | \J^{(k)}
  \Biggr)\\
&=\sum_{A\in \A(\J^{(k)})} \frac{Q_0(B\cap A)}{Q_0(A)} E\Biggl( \frac{Q^{(k)}(A)}{Q^{(k)}(A^p)}Q^{(k)}(A^p) \Big | \J^{(k)}
  \Biggr)\\
&=\sum_{A\in \A(\J^{(k)})}  \frac{Q_0(B\cap A)}{Q_0(A)}
  \frac{Q_0(A)}{Q_0(A^p)} E\Biggl( Q^{(k)}(A^p) \Big | \J^{(k)}
  \Biggr) \\
&=\sum_{A\in \A(\J^{(k)})}  \frac{Q_0(B\cap A)}{Q_0(A^p)}
   E\Biggl( Q^{(k-1)}(A^p) \Big | \J^{(k)}
  \Biggr)\\
&=\sum_{A\in \A(\J^{(k)})}Q_0(B\cap A)=Q_0(B).
\end{align*}
This shows that $E\left(Q^{(k)}(B)|\J^{(k)}\right)=Q_0(B)$ and thus
$EQ^{(k)}(B)=Q_0(B)$ for all $k$. But since
$|Q^{(k)}(B)-Q(B)| \rightarrow 0$ a.s.\!  (see the
proof of Theorem~1 in \cite{wongandma:2010}), by bounded convergence
theorem, we have
$E|Q^{(k)}(B)-Q(B)| \rightarrow 0$, and so $EQ(B)=Q_0(B)$.
\end{proof}

\begin{proof}[Proof of Theorem~2]
We first claim that with probability 1, $Q_1^{(k)}$ and $Q_2^{(k)}$
respectively converge in total variational distance to two absolutely
continuous random probability
measures $Q_1$ and $Q_2$, and thus for any Borel set $E$,
\begin{align*}
&|Q_1^{(k)}(E)-Q_1(E)| + |Q_2^{(k)}(E)-Q_2(E)| \\
\leq & \, sup_{E_1 \in
  \B}|Q_1^{(k)}(E_1)-Q_1(E_1)| + sup_{E_2 \in
  \B}|Q_1^{(k)}(E_2)-Q_1(E_2)|\rightarrow 0, \hbox{ w.p.1.}
\end{align*}
To prove the claim, we note that the marginal procedure that generates
$Q_1$, for instance, is simply an OPT with random local base measures
that arise from standard OPT distributions. To see this, we can think
of the generative procedure of $Q_1$ as consisting of the following
two steps.  
\begin{enumerate}
\item For each potential tree node $A$ under $\R$, we draw an
  independent random measure $Q_0^{A}$ from $OPT_{|A}(\R, \rho,
  \blam^b, \balp^b)$.  
\item Generate $Q_1$ from the standard
  random-partitioning-and-random-assignment procedure for an OPT,
  treating $\{C(A)\}$ as the stopping variables, $\{J(A)\}$ as the
  partition selector variables, and $\{\bthe^{J(A)}_1(A)\}$ as the
  probability assignment variables, and with $\{Q_0^{A}\}$ being the
  local base measures. That is, when a node $A$ is stopped, the conditional distribution $Q_1(\cdot |A)$ is set to be $Q_0^{A}(\cdot)$. 
\end{enumerate}
By Theorem~1 in \cite{wongandma:2010}, for each potential node $A$,
with probability 1, $Q_0^{A}$ is an absolutely continuous
distribution. Because the collection of all potential tree nodes $A$
under $\R$ is countable, with probability 1, this simultaneously holds
for all $Q_0^{A}$. Therefore, with probability 1, the marginal
procedure for producing $Q_1$ is just that for an OPT with
local base measures $\{Q_0^A\}$. The same argument for proving Theorem~1 in \cite{wongandma:2010} (with $\mu(\cdot |A)$ replaced by $Q_0^A(\cdot)$) shows that with probability 1, an absolutely continuous measure $Q_1$ exists as the limit of $Q_1^{(k)}$ in total variational distance. The same argument proves the claim for $Q_2$ as well.
 \end{proof}

\begin{proof}[Proof of Theorem~3]
Because any density function on $\om$ can be arbitrarily approximated
in $L_1$ by uniformly continuous ones, without loss of
generality, we can assume that $f_1$ and $f_2$ are uniformly continuous.
Let 
\[
\delta_1(\epsilon)=\sup_{|x-y|<\epsilon}|f_1(x)-f_1(y)| \quad
\text{and} \quad \delta_2(\epsilon)=\sup_{|x-y|<\epsilon}|f_2(x)-f_2(y)|.
\]
By uniform continuity, we have $\delta_i(\epsilon) \downarrow0$ as
$\epsilon\downarrow 0$ for $i=1,2$. Also, by Condition (1), for any
$\epsilon>0$, there exists a partition of $\om=\cup_{i=1}^{I} A_i$ such that the diameter of each $A_i$ is less than $\epsilon$. By Condition (2), there is positive probability that this partition will arise after a finite number of steps of recursive partitioning. Also because the parameters of the
co-OPT are all bounded away from 0 and 1, there is a positive
probability that the $A_i$'s are exactly the sets on which $Q_1$ and
$Q_2$ first couple. Now let $q^{A_i}$ be the local base measure on
each of $A_i$, we can write
\[
q_1(x) = \sum_{i=1}^{I} Q_1(A_i)q^{A_i}(x)\I_{A_i}(x) \quad \text{and}
\quad q_2(x) = \sum_{i=1}^{I} Q_2(A_i)q^{A_i}(x)\I_{A_i}(x).
\]
Accordingly,
\begin{align*}
&\,\,\,\,\,\int |q_1(x) - f_1(x)| d\mu(x) \\
&=\sum_{i=1}^{I}\int_{A_i}|Q_1(A_i)q^{A_i}(x) - f_1(x)| d\mu(x) \\
&\leq \sum_{i=1}^{I}Q_1(A_i)\int_{A_i}|q^{A_i}(x)-1/\mu(A_i)| d\mu(x) +
\sum_{i=1}^{I}\int_{A_i}|Q_1(A_i)/\mu(A_i) - f_1(x)| d\mu(x)\\
&\leq  \sum_{i=1}^{I} \int_{A_i}|q^{A_i}(x)\!-\!1/\mu(A_i)| d\mu(x)\!+\!
\sum_{i=1}^{I}|Q_1(A_i)\!-\! f_1^{i}\,\mu(A_i)| \!+\!
\sum_{i=1}^{I}\int_{A_i}|f_1^{i}\!-\! f_1(x)| d\mu(x)
\end{align*}
where $f_1^{j}:=\int_{A_i}f_1(x) d\mu(x) /\mu(A_i)$. 
By the exact same calculation we have
\begin{align*}
&\,\,\,\,\,\int |q_2(x) - f_2(x)| d\mu(x) \\
&\leq  \sum_{i=1}^{I} \int_{A_i}|q^{A_i}(x)\!-\! 1/\mu(A_i)| d\mu(x) \!+\!
\sum_{i=1}^{I}|Q_2(A_i) \!-\! f_2^{i}\,\mu(A_i)| + \sum_{i=1}^{I}\int_{A_i}|f_2^{i} \!-\! f_2(x)| d\mu(x)
\end{align*}
where $f_2^{j}:=\int_{A_i}f_2(x) d\mu(x) /\mu(A_i)$.
By the choice of
$A_i$, we have that 
$\int_{A_i}|f_1^{i} - f_1(x)| d\mu(x) \leq \delta_1(\epsilon)\mu(A_i)$
and $\int_{A_i}|f_2^{i} - f_2(x)| d\mu(x) \leq
\delta_2(\epsilon)\mu(A_i)$. Thus,
\[
\sum_{i=1}^{I}\int_{A_i}|f_1^{i} - f_1(x)| d\mu(x) \leq
\delta_1(\epsilon)\mu(\om) \quad \text{and} \quad \sum_{i=1}^{I}\int_{A_i}|f_2^{i} - f_2(x)| d\mu(x) \leq
\delta_2(\epsilon)\mu(\om).
\] 
So by choosing $\epsilon$ small enough, we can have
\[ \max\{\delta_1(\epsilon),\delta_2(\epsilon)\}\mu(\om) < \tau/3.\]
Next, because all the coupling parameters
of the co-OPT prior are uniformly bounded away from 0 and 1, (conditional on the
coupling partition) with positive
probability, we have
\[
|Q_1(A_i) - f_1^{i}\,\mu(A_i)| < \frac{\tau}{3\mu(\om)} \quad \text{and} \quad
|Q_2(A_i) - f_2^{i}\,\mu(A_i)| < \frac{\tau}{3\mu(\om)} 
\]
for all $i=1,2,\ldots,I$. Similarly, because all the base parameters
are also uniformly bounded away from 0 and 1, by Theorem~2 in
\cite{wongandma:2010}, (conditional on the 
coupling partition and probability assignments,) 
with positive probability we have 
\[
\int_{A_i}|q^{A_i}(x)-1/\mu(A_i)| d\mu(x) < \frac{\tau}{3\cdot 2^{i}}
\]
for all $i=1,2,\ldots,I$. Placing the three pieces together, we have
positive probability for $\int|q_1(x)-f_1(x)|d\mu<\tau$ and
$\int|q_2(x)-f_2(x)|d\mu<\tau$ to hold simultaneously.
\end{proof}

\begin{proof}[Proof of Proposition~5]
We prove the result only for $d_{L_1}$ as the proof for $d_{H^2}$ is very
similar. (All following equalities and statements hold with
probability 1.)
\begin{align*}
d_{L_1}(Q_1,Q_2) &= \int_{\om} |q_1(x)-q_2(x)| \mu(dx)\\
&=\sum_{A\in \A(C,J)}\int_{A}|q_1(x)-q_2(x)| \mu(dx) + \int_{\om
  \setminus \, \cup \A(C,J)}|q_1(x)-q_2(x)| \mu(dx).
\end{align*}
But for each $A\in \A(C,J)$, due to coupling we have $q_1(\cdot |A)=q_2(\cdot |A)$, and so
\begin{align*}
\int_{A}|q_1(x)-q_2(x)|\mu(dx) &= \int_{A} |Q_1(A)-Q_2(A)| q_1(x|A)
\mu(dx)\\
&=|Q_1(A)-Q_2(A)|.
\end{align*}
On the other hand, $Q_1(\om \setminus \cup\A(C,J))=Q_2(\om \setminus
\cup\A(C,J))=\mu(\om \setminus \cup\A(C,J))=0$ w.p.1. (See proof of
Theorem~1 in \cite{wongandma:2010}.) Therefore,
\[
d_{L_1}=\sum_{A \in \A(C,J)} |Q_1(A)-Q_2(A)|.
\]
\end{proof}

\section*{Appendix A.2. Prior specifications for Example~4}
For the P\'olya tree two sample test \cite{holmes:2009}, we have
imposed that each tree node is partitioned  
in the middle of both dimensions at each level. Therefore for our
example in $\real^2$, each node has four children. We also impose that the prior pseudo-counts
$\alpha$ are 0.5 for all children. The software used in this paper for
this method is written by us.

On the other hand, we used R package {\tt DPpackage} function {\tt
  HDPMdensity} to fit the Dirichlet Process mixture (DPM) model proposed in \cite{muller2004}. More
specifically, the two distributions are
modeled as.
\begin{align*}
F_1 &= \epsilon H_0 + (1-\epsilon) H_1 \\
F_2 &= \epsilon H_0 + (1-\epsilon) H_2,
\end{align*}
where $H_0$ models the common part of $F_1$ and $F_2$, whereas $H_1$
and $H_2$ the unique parts. The parameter $\epsilon$ captures the
proportion of``commonnes'' between the two distributions, and thus can
serve as a measure of how the two differ. Each of the $H_i$ for
$i=0,1,2$ is modeled as a Dirichlet Process mixture of normals.
\begin{align*}
H_i(\cdot) = \int \phi(\cdot| \mu,\Sigma) d G_i(\mu),
\end{align*}
where
\[
G_i | \alpha_i, G_0 \sim DP(\alpha_i,G_0).
\]
The baseline distribution $G_0$ is assumed to be
$Normal(\mu_0,\Sigma_0)$. Following the example given by {\tt
  DPpackage}, the (empirical) hyperprior specifications are
\begin{align*}
\epsilon &\sim 0.1\delta_0 + 0.1 \delta_1 + 0.8\, Unif[0,1],\\
\alpha_i &\sim Unif(0,1) \text{ for $i=0,1,2$.}\\
\Sigma_0|\mu_0,T_0 &\sim InverseWishart(\mu_0=9,T_0=\Var(\by))\\
\mu_0|m_0,S_0 &\sim N(m_0=mean(\by),S_0=\Var(\by))\\
\Sigma | \nu, T &\sim InverseWishart(\nu = 9, T=0.25\Var(\by)),
\end{align*}
where $\by$ is the combination of the two samples, $mean(\cdot)$ is
the dimension-wise average, and $\Var$ is the covariance. The
statistic we use to measure two sample difference (or similarity) is
the posterior mean of $\epsilon$, estimated by the mean of the MCMC
sample of size 10,000, with 10,000 burn-in steps.  

\section*{Appendix A.3. The dependent Dirichlet prior in Example~5}
Motivated by the hierarchical Dirichlet process mixture
prior setup introduced in \cite{muller2004}, we can design the
following prior for $(Q_1,Q_2)$ on the finite support of a contingency
table.
\[\left\{ \begin{array}{l}
Q_1 = \epsilon H_0 + (1-\epsilon) H_1\\
Q_2 = \epsilon H_0 + (1-\epsilon) H_2
\end{array}\right.
\]
with 
\begin{align*}
H_0, H_1, H_2 &\sim_{i.i.d} Dirichlet({\bm \alpha_H})\\
\epsilon &\sim Beta(a_{\epsilon},b_{\epsilon}).
\end{align*}
We used ${\bm \alpha_H}=(0.5,0.5,\ldots,0.5)$ and
$a_{\epsilon}=b_{\epsilon}=3$ as the prior parameters. We found that
restricting the support of $\alpha_H$ to the observed table cells
rather than the entire table significantly improves the power of the
method. This is due to the sparsity of the table counts---the vast
majority of the table cells are empty.

To draw posterior samples of $\epsilon$, we use the following Gibbs
sampler. First some notations. Let ${\bm X^1}=\{X^1_1,X^1_2,\ldots,
X^1_{n_1}\}$ and ${\bm X^2}=\{X^2_1,X^2_2,\ldots, X^2_{n_2}\}$ denote the two
sample observations. For each observation $X^i_j$ in sample $i=1$ or 2, we
introduce a Bernoulli variable $J^i_j$ that serves as an indicator for
whether $X_i^j$ has come for $H_i$ or $H_0$. Given $\epsilon$, the
$J^i_j$'s are i.i.d. $Bernoulli(\epsilon)$ variables. For simplicity,
we denote $(J^i_1,J^i_2,\ldots, J^i_{n_i})$ as ${\bm J^i}$. Given
$\bm{J}^i$ we let 
\[
\bX^{i,0} = \{X^i_j: J^i_j = 0, j=1,2,\ldots,n_i\} \quad \text{ and }
\quad \bX^{i,1} = \{X^i_j: J^i_j = 1, j=1,2,\ldots,n_i\}
\]
for $i=1,2$. In addition, we let $\bn(\bX^i)$ be the table counts of
of sample $i$ in the support of $\balp_H$, and similarly define
$\bn(\bX^{i,0})$ and $\bn(\bX^{i,1})$. With these notations, now
we next write
down the conditional distributions of $H_0$, $H_1$, $H_2$, $\epsilon$,
$\bj^1$ and $\bj^2$. 
\begin{align*}
H_0|\bX^1,\bX^2,\bj^1,\bj^2,\epsilon,H_1,H_2 &\sim Dirichlet(\balp_H
+ \bn(\bX^{1,0})+\bn(\bX^{2,0}))\\
H_1|\bX^1,\bX^2,\bj^1,\bj^2,\epsilon,H_0,H_2 &\sim Dirichlet(\balp_H
+ \bn(\bX^{1,1}))\\
H_2|\bX^1,\bX^2,\bj^1,\bj^2,\epsilon,H_0,H_1 &\sim Dirichlet(\balp_H
+ \bn(\bX^{2,1}))\\
J^1_j | \bX^1,\bX^2,\bj^1_{(-j)},\bj^2,\epsilon,H_0,H_1,H_2 &\sim
Bernoulli\left(\frac{(1-\epsilon)p_{H_1}(X^1_j)}{\epsilon
    p_{H_0}(X^1_j)+(1-\epsilon)p_{H_1}(X^1_j)} \right)\\
J^2_j | \bX^1,\bX^2,\bj^1,\bj^2_{(-j)},\epsilon,H_0,H_1,H_2 &\sim
Bernoulli\left(\frac{(1-\epsilon)p_{H_1}(X^1_j)}{\epsilon
    p_{H_0}(X^1_j)+(1-\epsilon)p_{H_1}(X^1_j)} \right)\\
\epsilon|\bX^1,\bX^2,\bj^1,\bj^2,H_0,H_1,H_2 &\sim Beta\left(a_\epsilon +
\sum_{i=1}^{2}\sum_{j=1}^{n_i} J^i_j,\,\, b_\epsilon + n_1 +
n_2 - \sum_{i=1}^{2}\sum_{j=1}^{n_i} J^i_j\right).
\end{align*}
We use this Gibbs sampler to draw posterior samples for $\epsilon$. We compute the posterior mean of $\epsilon$ from $10,000$ samples with $10,000$ burn-in iterations.

\bibliography{coupled_opt}

\end{document}